\documentclass{focs}

\usepackage{algorithm}
\usepackage{algorithmic}
\usepackage{amsmath}
\usepackage{amsthm}
\usepackage{amsfonts}
\usepackage{paralist}
\usepackage{xspace}
\usepackage{marginnote}
\usepackage{times}
\usepackage{fullpage}
\usepackage{ifpdf}

\ifpdf \usepackage[pdftex,%
pdfauthor={Nicholas J.A. Harvey, Jan Vondrak},%
pdftitle={An Algorithmic Proof of the Lovasz Local Lemma via Resampling Oracles},%
colorlinks=true,%
linkcolor=blue,%
citecolor=blue,%
hypertexnames=true,%
bookmarks=true,%
pagebackref]{hyperref}
\else
\usepackage{hyperref}
\fi

\newtheorem{theorem}{Theorem}[section]
\newtheorem*{theorem*}{Theorem}
\newtheorem{lemma}[theorem]{Lemma}

\newtheorem{corollary}[theorem]{Corollary}

\newtheorem{claim}[theorem]{Claim}
\newtheorem{definition}[theorem]{Definition}

\reversemarginpar

\newcommand{\partdiff}[2]{\frac{\partial {#1}}{\partial {#2}}}

\newcommand{\mixdiff}[3]{\frac{\partial^2 {#1}}{{\partial {#2}}{\partial {#3}}}}
\newcommand{\set}[1]{\left \{ #1 \right \}}                     
\newcommand{\setst}[2]{\left\{\; #1 \,:\, #2 \;\right\}}        
\newcommand{\union}{\cup}
\newcommand{\intersect}{\cap}
\newcommand{\card}[1]{\lvert #1 \rvert}
\newcommand{\smallsum}[2]{{\textstyle \sum_{#1}^{#2}}}

\newcommand{\Algorithm}[1]{Algorithm~\ref{alg:#1}}
\newcommand{\AlgorithmName}[1]{\label{alg:#1}}
\newcommand{\Appendix}[1]{Appendix~\ref{app:#1}}
\newcommand{\AppendixName}[1]{\label{app:#1}}
\newcommand{\Claim}[1]{Claim~\ref{clm:#1}}
\newcommand{\ClaimName}[1]{\label{clm:#1}}
\newcommand{\Corollary}[1]{Corollary~\ref{cor:#1}}
\newcommand{\CorollaryName}[1]{\label{cor:#1}}
\newcommand{\Definition}[1]{Definition~\ref{def:#1}}
\newcommand{\DefinitionName}[1]{\label{def:#1}}
\newcommand{\Equation}[1]{\eqref{eq:#1}}
\newcommand{\EquationName}[1]{\label{eq:#1}}
\newcommand{\Lemma}[1]{Lemma~\ref{lem:#1}}
\newcommand{\LemmaName}[1]{\label{lem:#1}}
\newcommand{\Section}[1]{Section~\ref{sec:#1}}
\newcommand{\SectionName}[1]{\label{sec:#1}}
\newcommand{\Theorem}[1]{Theorem~\ref{thm:#1}}
\newcommand{\TheoremName}[1]{\label{thm:#1}}

\newcommand{\proofbelow}{5pt}
\newcommand{\afterproof}{\hfill $\square$ \par \vspace{\proofbelow}}
\newenvironment{proofof}[1]{\noindent\textit{Proof} \,(of #1).\,}{\afterproof}

\def\E{{\bf E}}
\def\b1{{\bf 1}}

\def\RR{{\mathbb R}}

\def\cE{{\cal E}}
\def\cF{{\cal F}}
\def\cI{{\cal I}}
\def\cJ{{\cal J}}
\def\cS{{\cal S}}
\def\cU{{\cal U}}
\def\Ind{{\sf Ind}}
\def\Prop{{\sf Prop}}
\def\Stab{{\sf Stab}}
\def\poly{{\mbox{poly}}}

\newcommand{\newterm}[1]{\textit{#1}}

\renewcommand{\th}{\ifmmode{^{\textrm{th}}}\else{\textsuperscript{th}\ }\fi}

\newcommand{\sumstack}[1]{\sum_{\substack{#1}}}

\newcommand{\NDG}{lopsidependency graph\xspace}

\newcommand{\NAG}{lopsided association graph\xspace}
\newcommand{\NAP}{lopsided association\xspace}
\newcommand{\ro}{resampling oracle\xspace}
\newcommand{\ros}{resampling oracles\xspace}
\newcommand{\Ro}{Resampling oracle\xspace}
\newcommand{\Ros}{Resampling oracles\xspace}
\newcommand{\MSR}{MaximalSet\-Resample\xspace}

\newcommand{\RONE}{{\rm(R1)}}
\newcommand{\RTWO}{{\rm(R2)}}

\title{An Algorithmic Proof of the Lov\'asz Local Lemma \\ via Resampling Oracles}
\date{}
\numberofauthors{2}
\author{
\alignauthor
Nicholas J.~A.~Harvey\\
       \affaddr{University of British Columbia}\\
       \affaddr{Vancouver, Canada}\\
       \email{nickhar@cs.ubc.ca}
\alignauthor
Jan Vondr\'{a}k\\
\affaddr{IBM Almaden Research Center}\\
       \affaddr{San Jose, CA, USA}\\
       \email{jvondrak@us.ibm.com}
       \alignauthor
}
\date{\today}

\begin{document}

\pagestyle{empty}

\maketitle

\begin{abstract}
The Lov\'asz Local Lemma is a seminal result in probabilistic combinatorics.
It gives a sufficient condition on a probability space and a collection of events
for the existence of an outcome that simultaneously avoids all of those events.
Finding such an outcome by an efficient algorithm has been an active research topic for
decades.
Breakthrough work of Moser and Tardos (2009) presented an efficient algorithm for a general setting
primarily characterized by a product structure on the probability space.

In this work we present an efficient algorithm for a much more general setting.
Our main assumption is that there exist certain functions,
called \textit{resampling oracles},
that can be invoked to address the undesired occurrence of the events.
We show that, in \emph{all} scenarios to which the original Lov\'asz Local Lemma applies,
there exist resampling oracles, although they are not necessarily efficient.
Nevertheless, for essentially all known applications of the Lov\'asz Local Lemma
and its generalizations,
we have designed efficient resampling oracles.
As applications of these techniques, we present new results for packings of Latin transversals,
rainbow matchings and rainbow spanning trees.
\end{abstract}

\newpage \pagestyle{plain}\setcounter{page}{1}

\clearpage
\tableofcontents
\clearpage

\section{Introduction}

The Lov\'asz Local Lemma (LLL) is a powerful tool with numerous uses
in combinatorics and theoretical computer science.
If a given probability space and collection of events satisfy a certain condition,
then the LLL asserts the existence of an outcome that simultaneously avoids those events.
The classical formulation of the LLL \cite{ErdosLovasz,Spencer77} is as follows.

Let $\Omega$ be a probability space with probability measure $\mu$.
Let $E_1,\ldots,E_n$ be certain ``undesired'' events in that space.
Let $G$ be an undirected graph with vertex set $[n]=\set{1,\ldots,n}$.
The edges of $G$ are denoted $E(G)$.
Let $\Gamma(i) = \setst{ j \neq i }{ \set{i,j} \in E(G) }$ be the neighbors of vertex $i$.
Also, let $\Gamma^+(i) = \Gamma(i) \union \set{i}$ and
let $\Gamma^+(I) = \bigcup_{i \in I} \Gamma^+(i)$ for $I \subseteq [n]$.

\begin{theorem}[General Lov\'asz Local Lemma~\cite{ErdosLovasz,Spencer77}]
\TheoremName{LLL}
Suppose that the events satisfy the following 
condition that controls their dependences
\begin{equation}
\tag{Dep}
\label{eq:Dep}
\Pr_\mu[E_i \mid \cap_{j \in J} \overline{E_j}] ~=~ \Pr_\mu[E_i]
    \qquad\forall i \in [n] ,\, J \subseteq [n] \setminus \Gamma^+(i)
\end{equation}
and the following criterion that controls their probabilities
\begin{equation}
\tag{GLL}
\EquationName{GLL}
\exists x_1,\ldots,x_n \in (0,1) 
\qquad\text{such that}\qquad
\Pr_\mu[E_i] ~\leq~ x_i \prod_{j \in \Gamma(i)} (1-x_j)
~\quad\forall i \in [n].
\end{equation}
Then $\Pr_\mu[\bigcap_{i=1}^{n} \overline{E_i}]
> 0$.
\end{theorem}

An equivalent statement of \eqref{eq:Dep} is that the event $E_i$ 
must be independent of the joint distribution on the events
$\setst{ E_j }{ j \not\in \Gamma^+(i) }$.
When \eqref{eq:Dep} holds, $G$ is called a \emph{dependency graph}.
The literature contains several dependency conditions generalizing \eqref{eq:Dep} and
criteria generalizing \eqref{eq:GLL} under which the conclusion of the
theorem remains true. We will discuss several such generalizations below.

The LLL can also be formulated \cite{AlonSpencer} in terms of a directed dependency graph 
instead of an undirected graph, but nearly 
all applications of which we are aware involve an undirected graph.
Accordingly, our work focuses primarily on the undirected case,
but we will mention below which of our results extend to the directed case.

\paragraph{Algorithms.}
Algorithms to efficiently find an outcome
in $\bigcap_{i=1}^{n} \overline{E_i}$ have been the subject of research
for several decades.
In 2008, a nearly optimal result was obtained by
Moser \cite{Moser} for a canonical application of the LLL, the bounded-degree $k$-SAT problem.
Shortly thereafter, Moser and Tardos \cite{MoserTardos} extended that result to a 
general scenario called the ``variable model'' in which $\Omega$ consists of independent variables,
each $E_i$ depends on a subset of the variables,
and events $E_i$ and $E_j$ are adjacent in $G$ if there is a variable on which they both depend.
Clearly the resulting graph is a dependency graph.
The Moser-Tardos algorithm is extremely simple:
after drawing an initial sample of the variables, it repeatedly checks if any
undesired event occurs, then \textit{resamples} any such event.
Resampling an event means that the variables on which it depends receive fresh samples according to
$\mu$.  Moser and Tardos prove that, if the \eqref{eq:GLL} condition is satisfied,
this algorithm will produced the desired outcome after at most $\sum_{i=1}^n \frac{x_i}{1-x_i}$ 
resampling operations, in expectation.

Numerous extensions of the Moser-Tardos algorithm have been proposed.
These extensions can handle more general criteria \cite{Kolipaka,Pegden,AG,KSX},
derandomization \cite{CGH}, exponentially many events \cite{HSS},
distributed scenarios \cite{Pettie}, etc.
However, these results are restricted to the Moser-Tardos variable model and
hence cannot be viewed as algorithmic proofs of the LLL in full generality.
There are many known scenarios for the LLL and its generalizations that fall outside the scope
of the variable model \cite{LuMohrSzekely,Mohr-thesis}.
\Section{implementation} discusses several such scenarios,
including random permutations, matchings and spanning trees.

Recently two efficient algorithms have been developed that go beyond the variable model.
Harris and Srinivasan \cite{HarrisS14} extend the Moser-Tardos algorithm to a scenario 
involving random permutations that originates in work of Erd\H{o}s and Spencer \cite{ErdosSpencer}.
Achlioptas and Iliopoulos \cite{Achlioptas} developed a novel algorithmic ``flaw correction'' framework 
which allows one to model various applications of the LLL in a flexible manner.
They show how this captures several applications of the LLL outside the variable model, and even
some results that might be beyond typical formulations of the LLL.
In contrast to the other results mentioned here, their framework does not involve an underlying
measure $\mu$ and is not directly tied to the probabilistic setting of the LLL.
This has some benefits, but also some restrictions that seem to prevent it from recovering the LLL
in full generality,
In particular, their publication \cite{Achlioptas} does not claim a formal connection
with \Theorem{LLL}.
\Section{related-work} contains further discussion of the related work.

\subsection{Our contributions}

The primary motivating question for this work is whether there is an ``algorithmic proof" of the Lov\'asz Local Lemma in general probability spaces. We answer this question in the following sense:
We propose an algorithmic framework for the general Lov\'asz Local Lemma,
based on a new notion of \newterm{resampling oracles}. In this framework, 
we present an algorithm that finds a point in $\bigcap_{i=1}^{n} \overline{E_i}$ (avoiding all undesired events) efficiently, if given access to three types of subroutines outlined below (the most crucial one being resampling oracles). Whether these subroutines can be implemented efficiently is an instance-dependent issue, and we discuss this further below. However, we show that the existence of such subroutines is guaranteed by the assumptions of the Lov\'asz Local Lemma. In particular, our algorithm provides a new proof of \Theorem{LLL} (with no further assumptions), and several generalizations thereof, as described below.
Algorithmically, we reduce the problem of finding a point in $\bigcap_{i=1}^{n} \overline{E_i}$ to
the problem of implementing the three subroutines that we discuss next.

\subsubsection{Algorithmic assumptions}
\SectionName{algass}

In order to discuss algorithms for the LLL in full generality, one must assume some form of access to the probability space at hand.
It is natural to assume that one can efficiently sample from $\mu$, and efficiently check whether
a given event $E_i$ occurs. However, even under these assumptions,
finding the desired output can be computationally hard.
(We show an example demonstrating this in \Section{hardness}.)
Therefore, our framework assumes the existence of one more subroutine that can
be used by our algorithm. This leads us to the notion of \ros.

Let us introduce some notation.
An atomic event $\omega$ in the probability space $\Omega$ will be called a \newterm{state}.
We write $\omega \sim \mu$ to denote that a random state $\omega$ is distributed according to $\mu$,
and $\omega \sim \mu |_{E_i}$ to denote that the distribution is $\mu$ conditioned on $E_i$.
The \ros are defined with respect to a graph $G$ on $[n]$ with neighborhood structure $\Gamma$
(not necessarily satisfying the \eqref{eq:Dep} condition).

The three subroutines required by our algorithm are as follows.
\begin{itemize}
\item\textit{Sampling from $\mu$:} There is a subroutine that provides an independent random state $\omega \sim \mu$.

\item\textit{Checking events:}
For each $i \in [n]$, there is a subroutine that determines whether $\omega \in E_i$.

\item\textit{\Ros:} For each $i \in [n]$, there is a randomized subroutine
$r_i:\Omega \rightarrow \Omega$ with the following properties.
\begin{description}
\item[\rm \RONE] If $E_i$ is an event and $\omega \sim \mu|_{E_i}$, then $r_i(\omega) \sim \mu$.
(The oracle $r_i$ removes conditioning on $E_i$.)
\item[\rm \RTWO] For any $j \notin \Gamma^+(i)$, if $\omega \not\in E_j$
then also $r_i(\omega) \not\in E_j$.
(Resampling an event cannot cause new non-neighbor events to occur.)
\end{description}
When these conditions hold, we say that $r_i$ is a \ro for events $E_1,\ldots,E_n$
and graph $G$. 
\end{itemize}

If efficiency concerns are ignored, the first two subroutines trivially exist.
We show that (possibly inefficient) resampling oracles exist if and only if a certain relaxation of \eqref{eq:Dep} holds (see \Section{lopsided-intro}).

\medskip
\noindent
{\bf Main Result.}
Our main result is that we can find a point in $\bigcap_{i=1}^n \overline{E_i}$ efficiently, whenever the three subroutines above have efficient implementations. 

\begin{theorem*}[Informal] 
Consider any probability space, any events $E_1,\ldots,E_n$,
and any undirected graph $G$ on vertex set $[n]$.
If \eqref{eq:GLL} is satisfied and if the three subroutines described above are available,
then our algorithm finds a state in $\bigcap_{i=1}^n \overline{E_i}$ efficiently in terms of
the number of calls to these subroutines.
\end{theorem*}

We make a more precise statement in the following section.
We note that this theorem does not assume that \eqref{eq:Dep} holds, 
and the existence of \ros is actually a strictly weaker condition.
Thus, our algorithm provides a new proof of \Theorem{LLL} (the existential LLL) under its original assumptions.

\subsection{Our algorithm: \MSR}
\SectionName{MSR}

A striking aspect of the work of Moser and Tardos \cite{MoserTardos} is the simplicity and
flexibility of their algorithm --- in each iteration, \emph{any} event $E_i$ that occurs
can be resampled.
We propose a different algorithm that is somewhat less flexible, but whose
analysis seems to be simpler in our scenario.
Roughly speaking, our algorithm proceeds in iterations where in each iteration we resample
events that form an independent set in $G$. The independent set is generated by a greedy algorithm
that adds a vertex $i$ and resamples $E_i$, if $i$ is not adjacent to the previously selected vertices
and $E_i$ occurs in the current state. This is repeated until no events occur.
Pseudocode for this procedure is shown in \Algorithm{resample}.
Nearly identical algorithms have been proposed before,
particularly parallel algorithms \cite{MoserTardos,Kolipaka},
although our interest lies not in the parallel aspects but rather in making the LLL
(and its stronger variants) algorithmic in our general setting.

\begin{algorithm}
\caption{\MSR uses \ros to output a state $\omega \in \bigcap_{i=1}^n \overline{E_i}$.
It requires the three subroutines described in \Section{algass}:
    sampling $\omega \sim \mu$,
    checking if an event $E_i$ occurs,
    and the resampling oracles $r_i$.
}
\AlgorithmName{resample}
\begin{algorithmic}[1]
\STATE Initialize $\omega$ with a random state sampled from $\mu$;
\STATE $t := 0$;
\REPEAT
\STATE $t := t+1$;
\STATE $J_t := \emptyset$
\WHILE {there is $i \notin \Gamma^+(J_t)$ such that $\omega \in E_i$}
\STATE Let $i$ be the minimum index satisfying that condition;
\STATE $J_t := J_t \cup \{i\}$;
\STATE $\omega := r_i(\omega)$;     \quad$\rhd$ \textit{Resample $E_i$}
\ENDWHILE
\UNTIL {$J_t = \emptyset$};
\RETURN $\omega$.
\end{algorithmic}
\end{algorithm}

Our algorithmic proof of the LLL amounts to showing that \MSR terminates,
at which point $\omega \in \bigcap_{i=1}^n \overline{E_i}$ clearly holds.
Our bound on the running time of \MSR is shown by the following theorem, which is proven in \Section{analysis}. We note that our bound is at most quadratic in the quantity $\sum_{i=1}^{n} \frac{x_i}{1-x_i}$ which was the bound proved by Moser and Tardos \cite{MoserTardos}.

\begin{theorem}
\TheoremName{LLL-tight-result}
Suppose that the events $E_1,\ldots,E_n$ satisfy \eqref{eq:GLL}
and that the three subroutines described above in \Section{algass} are available.
Then the expected number of calls to the \ros before \MSR terminates is 
$O\big(\sum_{i=1}^{n} \frac{x_i}{1-x_i} \sum_{j=1}^{n} \log \frac{1}{1-x_j} \big)$.
\end{theorem}

\subsection{Generalizing the dependency condition}
\SectionName{lopsided-intro}

Erd\H{o}s and Spencer \cite{ErdosSpencer} showed that \Theorem{LLL}
still holds when \eqref{eq:Dep} is generalized to\footnote{
    More precisely, \eqref{eq:LLLL} should be restricted to $J$ for which
    $\Pr_\mu[\cap_{j \in J} \overline{E_j}] > 0$.
    However that restriction is ultimately unnecessary because, in the context of the LLL,
    the theorem of Erd\H{o}s and Spencer implies that $\Pr_\mu[\cap_{j \in [n]} \overline{E_j}] > 0$.
}
\begin{equation}
\tag{Lop}
\label{eq:LLLL}
\Pr_\mu[E_i \mid \cap_{j \in J} \overline{E_j}] ~\leq~ \Pr_\mu[E_i]
    \qquad\forall i \in [n] ,\, J \subseteq [n] \setminus \Gamma^+(i).
\end{equation}
They playfully called this the ``lopsidependency'' condition,
and called $G$ a ``\NDG''.
This more general condition enables several interesting uses of the LLL 
in combinatorics and theoretical computer science,
e.g., existence of Latin transversals \cite{ErdosSpencer} and
optimal thresholds for satisfiability \cite{Gebauer}.

Recall that \Theorem{LLL-tight-result} did not assume \eqref{eq:Dep} and instead
assumed the existence of \ros.
It is natural to wonder how the latter assumption relates to lopsidependency.
We show that the existence of \ros is equivalent to a condition that we call
\newterm{lopsided association}, and whose strength lies strictly
between \eqref{eq:Dep} and \eqref{eq:LLLL}.
The lopsided association condition is
\begin{equation}
\tag{LopA}
\EquationName{Neg}
\Pr_\mu[E_i \cap F] ~\geq~ \Pr_\mu[E_i] \cdot \Pr_\mu[F]
    \qquad \forall i \in [n], \forall F \in {\cal F}_i
\end{equation}
where ${\cal F}_i$ contains all events $F$ whose indicator variable is a monotone non-decreasing function of the
indicator variables of $(E_j \,:\, j \notin \Gamma^+(i))$.
We call a graph satisfying \eqref{eq:Neg} a \newterm{\NAG}
for events $E_1,\ldots,E_n$.

\begin{theorem*}[Informal] 
\Ros exist for events $E_1,\ldots,E_n$ and a graph $G$ if and only if $G$ is a \NAG
for events $E_1,\ldots,E_n$.
\end{theorem*}

This equivalence follows essentially from LP duality: The existence of a \ro can be formulated as a
{\em transportation problem} for which the \NAP condition is exactly the necessary and sufficient
condition for a feasible transportation to exist. 
\Section{existence} proves this result in detail.

As remarked above, the dependency conditions are related by
\eqref{eq:Dep} $\Rightarrow$ \eqref{eq:Neg} $\Rightarrow$ \eqref{eq:LLLL}.
The first implication is obvious since \eqref{eq:Dep} implies that $E_i$ is independent of $F$ in \eqref{eq:Neg}.
To see the second implication, simply take
$F = \bigcup_{j \in J} E_j$ for any $J \subseteq [n] \setminus \Gamma^+(i)$
to obtain that 
$\Pr_\mu[E_i \mid \cup_{j \in J} E_j] \geq \Pr_\mu[E_i]$.
Although lopsided association is formally a stronger assumption than lopsidependency,
every use of the LLL with lopsidependency that we have studied actually satisfies
the stronger lopsided association condition.
We demonstrate this in \Section{implementation} by designing efficient \ros for those scenarios.
Consequently, \Theorem{LLL-tight-result} makes the LLL efficient in those scenarios.

As remarked above, \Section{hardness} describes a scenario in which
\eqref{eq:Dep} and \eqref{eq:GLL} are satisfied for a dependency graph $G$
but finding a state $\omega \in \bigcap_{i=1}^{n} \overline{E_i}$ is computationally hard,
assuming standard complexity theoretic beliefs.
In that scenario \ros must necessarily exist since \eqref{eq:Dep} is satisfied,
but they cannot be efficiently implemented due to the computational hardness.
Therefore the equivalence between \eqref{eq:Neg} and \ros
comes with no efficiency guarantees.
Nevertheless in all lopsidependency scenarios that we have encountered in applications of the LLL,
efficient implementations of the \ros arise naturally from existing work,
or can be devised with modest effort.
In particular this is the case for random permutations,
perfect matchings in complete graphs, and spanning trees in complete graphs,
as discussed in \Section{implementation}.

\subsection{Generalizing the LLL criterion}
\SectionName{generalizingLLL}

In the early papers on the LLL \cite{ErdosLovasz,Spencer77}, the \eqref{eq:GLL} criterion 
relating the dependency graph $G$ and the probabilities $\Pr_\mu[E_i]$ 
was shown to be a sufficient condition to ensure that
$\Pr_\mu[\bigcap_{i=1}^{n} \overline{E_i}] > 0$.
Shearer \cite{Shearer} discovered a more general criterion that ensures the same conclusion.
In fact, Shearer's criterion is the best possible: whenever his criterion is violated,
there exist a corresponding measure $\mu$ and events $E_1,\ldots,E_n$ 
for which $\Pr_\mu[\bigcap_{i=1}^{n} \overline{E_i}] = 0$.

\Section{analysis} formally defines Shearer's criterion and uses it in a fundamental
way to prove \Theorem{LLL-tight-result}.
Moreover, we give an algorithmic proof of the LLL under Shearer's criterion instead of
the \eqref{eq:GLL} criterion.
This algorithm is efficient in typical situations, although the efficiency depends on
Shearer's parameters.
The following simplified result
is stated formally and proven in \Section{shearer-automatic-slack}.

\begin{theorem*}[Informal]
Suppose that a graph $G$ and the probabilities $\Pr_\mu[E_1],\ldots,\Pr_\mu[E_n]$ satisfy
Shearer's criterion with $\epsilon$ slack,
and that the three subroutines described in \Section{algass} are available.
Then the expected number of calls to the \ros by \MSR is $O(\frac{n}{\epsilon} \log \frac{1}{\epsilon})$.
\end{theorem*}

We also prove a more refined bound valid for any probabilities satisfying Shearer's criterion.
This bound is similar to the bound obtained by Kolipaka and Szegedy \cite{Kolipaka}; see \Section{shearer-automatic-slack}
for details.

Unfortunately Shearer's criterion is unwieldy and has not seen much use in applications of the LLL.
Recently several researchers have proposed criteria of intermediate strength between \eqref{eq:GLL}
and Shearer's criterion \cite{Bissacot,KSX}.
The first of these, called the \newterm{cluster expansion} criterion, was originally devised by
Bissacot et al.~\cite{Bissacot}, and is based on insights from statistical physics.
This criterion has given improved results in several
applications of the local lemma \cite{Bottcher,HarrisS14,Ndreca}.
Previous algorithmic work has also used the cluster expansion criterion 
in the variable model \cite{AG,Pegden} and for permutations \cite{HarrisS14}.

We give a new, elementary proof that the cluster expansion criterion implies Shearer's criterion.
In contrast, the previous proof is analytic and requires several ideas from
statistical physics \cite{Bissacot}.
As a consequence, we obtain the first purely combinatorial proof that the existential LLL holds
under the cluster expansion criterion.
Another consequence (\Theorem{cluster-no-slack})
is an algorithm for the LLL under the cluster expansion criterion,
obtained using our algorithmic results under Shearer's criterion.
This generalizes \Theorem{LLL-tight-result} by replacing \eqref{eq:GLL}
with the cluster expansion criterion, stated below as \eqref{eq:CLL}.
To state the result, we require additional notation:
let $\Ind$ denote the family of independent sets in the graph $G$.

\begin{theorem}
\TheoremName{cluster-no-slack}
Suppose that the events $E_1,\ldots,E_n$ satisfy the following criterion
\begin{equation}
\tag{CLL}
\EquationName{CLL}
\exists y_1,\ldots,y_n>0
\qquad\text{such that}\qquad
\Pr_\mu[E_i] ~\leq~ \frac{ y_i }{ \sum_{J \subseteq \Gamma^+(i), J \in \Ind} \prod_{j \in J} y_j }.
\end{equation}
and that the three subroutines described in \Section{algass} are available.
Then the expected number of calls to the \ros before \MSR terminates is
$O\big(\sum_{i=1}^{n} y_i \sum_{j=1}^{n} \ln (1+y_j) \big)$.
\end{theorem}

\subsection{Techniques and related work}
\SectionName{related-work}

The breakthrough work of Moser and Tardos \cite{Moser,MoserTardos}
stimulated a string of results on algorithms for the LLL.
This section reviews the results that are most relevant to our work.
Several interesting techniques play a role in the analyses of these previous algorithms.
These can be roughly categorized as the
\emph{entropy method} \cite{Moser-thesis,Achlioptas},
\emph{witness trees} or \emph{witness sequences} \cite{MoserTardos,HarrisS14,Kolipaka}
and \emph{forward-looking combinatorial analysis} \cite{GiotisKPT}.

Moser \cite{Moser,Moser-thesis} developed the entropy method to
analyze a very simple algorithm for the ``symmetric'' LLL \cite{ErdosLovasz},
which incorporates the maximum degree of $G$ and a uniform bound on $\Pr_\mu[E_i]$.
The entropy method roughly shows that, if the algorithm runs for a long time, a transcript of the
algorithm's actions provides a compressed representation of the algorithm's random bits,
which is unlikely due to entropy considerations.

Following this, Moser and Tardos \cite{MoserTardos} showed that a similar algorithm
will produce a state in $\bigcap_{i=1}^n \overline{E_i}$,
assuming the independent variable model and the \eqref{eq:GLL} criterion.
This paper is primarily responsible for the development of witness trees, and proved the
``witness tree lemma'', which yields an extremely elegant analysis in the variable model.
The witness tree lemma has further implications. For example, it allows one to
analyze separately for each event its expected number of resamplings.
Moser and Tardos also extended the variable model to incorporate a limited form of lopsidependency,
and showed that their analysis still holds in that setting.

The main advantage of our result over the Moser-Tardos result is that we address the occurrence
of an event through the abstract notion of \ros rather than directly resampling
the variables of the variable model. Furthermore we give efficient implementations of 
\ros for essentially all known probability spaces to which the LLL has been applied.
A significant difference with our work is that we do not have an analogue of the witness tree lemma;
our approach provides a simpler analysis when the LLL criterion has slack but requires a more complicated analysis
to remove the slack assumption.
As a consequence, our bound on the number of \ro calls is larger than
the Moser-Tardos bound.
Our lack of a witness tree lemma is inherent.
\Appendix{witness-trees} shows that 
the witness tree lemma is false in the abstract scenario of \ros.

The Moser-Tardos algorithm is known to terminate under criteria more general than
\eqref{eq:GLL}, while still assuming the variable model.
Pegden \cite{Pegden} showed that the cluster expansion criterion suffices, whereas
Kolipaka and Szegedy \cite{Kolipaka} showed more generally that Shearer's criterion suffices.
We also extend our analysis to the cluster expansion criterion as well as Shearer's criterion,
in the more general context of \ros.
Our bounds on the number of resampling operations are somewhat weaker than those of
\cite{Pegden,Kolipaka}, but the increase is at most quadratic.

Kolipaka and Szegedy \cite{Kolipaka} present another algorithm, called GeneralizedResample,
whose analysis proves the LLL under Shearer's condition for arbitrary probability spaces.
GeneralizedResample is similar to \MSR in that they both work with abstract distributions and that
they repeatedly choose a maximal independent set $J$ of undesired events to resample.
However, the way that the bad events are resampled is different:
GeneralizedResample needs to sample from $\mu |_{\cap_{j \not\in \Gamma^+(J)} \overline{E_j}}$,
which is a complicated operation that seems difficult to implement efficiently.
Thus \MSR can be viewed as a variant of GeneralizedResample that can be made efficient in all known
scenarios.

Harris and Srinivasan \cite{HarrisS14} show that the Moser-Tardos algorithm
can be adapted to handle certain events in a probability space involving random permutations. 
Their method for resampling an event is based on the Fischer-Yates shuffle.
This scenario can also be handled by our framework;
their resampling method perfectly satisfies the criteria of a \ro.
The Harris-Srinivasan's result is stronger than ours in that they do prove an analog of the
witness tree lemma. Consequently their algorithm requires fewer resamplings than ours,
and they are able to derive parallel variants of their algorithm.
The work of Harris and Srinivasan is technically challenging, and generalizing it to a more abstract
setting seems daunting.

Achlioptas and Iliopoulos \cite{Achlioptas,AchlioptasI15} proposed a general framework for finding
``flawless objects'', based on actions for addressing flaws.
We call this the A-I framework.
They show that, under certain conditions,
a random walk over such actions rapidly converges to a flawless object.
This naturally relates to the LLL by viewing each event $E_i$ as a flaw.
At the same time, the A-I framework is not tied to the probabilistic formulation of the LLL,
and can derive results, such as the greedy algorithm for vertex coloring,
that seem to be outside the scope of typical LLL formulations, such as \Theorem{LLL}.
The A-I framework \cite{Achlioptas,AchlioptasI15} has other restrictions and does not claim
to recover any particular form of the LLL.
Nevertheless, the framework can accommodate applications of the LLL where lopsidependency
plays a role, such as rainbow matchings and rainbow Hamilton cycles.
In contrast, our framework embraces the probabilistic formulation
and can recover the original existential LLL (\Theorem{LLL}) in full generality,
even incorporating Shearer's generalization.
The A-I analysis \cite{Achlioptas} is inspired by Moser's entropy method. Technically, it entails
an encoding of random walks by ``witness forests" and combinatorial counting thereof
to estimate the length of the random walk. The terminology of witness forests is reminiscent
of the witness trees of Moser and Tardos, but conceptually they are different in that
the witness forests grow ``forward in time" rather than backward. 
This is conceptually similar to ``forward-looking combinatorial analysis", which we discuss next.

Giotis et al.\ \cite{GiotisKPT} show that a variant of Moser's algorithm
gives an algorithmic proof in the variable model of the symmetric LLL.
While this result is relatively limited when compared to the results above,
their analysis is a clear example of forward-looking combinatorial analysis.
Whereas Moser and Tardos use a \emph{backward-looking} argument to find witness trees
in the algorithm's ``log'', Giotis et al.\ analyze a \emph{forward-looking} structure:
the tree of resampled events and their dependencies, looking forward in time.
This viewpoint seems more natural and suitable for extensions.

Our approach can be roughly described as {\em forward-looking analysis} with a careful modification
of the Moser-Tardos algorithm, formulated in the framework of \ros.
Our main conceptual contribution is the simple definition of the \ros,
which allows the resamplings to be readily incorporated into the forward-looking analysis.
Our modification of the Moser-Tardos algorithm is designed to combine this analysis
with the technology of ``stable set sequences'' \cite{Kolipaka},
defined in \Section{stable-set-sequences},
which allows us to accommodate various LLL criteria, including Shearer's criterion.
This plays a fundamental role in the full proof of \Theorem{LLL-tight-result}.

Our second contribution is a technical idea concerning slack in the LLL criteria.
This idea is a perfectly valid statement regarding the existential LLL as well,
although we will exploit it algorithmically.
One drawback of the forward-looking analysis is that it naturally leads to an 
exponential bound on the number of resamplings, unless there is some slack in the 
LLL criterion; this same issue arises in \cite{Achlioptas,GiotisKPT}.
Our idea eliminates the need for slack in the \eqref{eq:GLL} and \eqref{eq:CLL} criteria.
We prove that, even if \eqref{eq:GLL} or \eqref{eq:CLL} are tight,
we can instead perform our analysis using Shearer's criterion,
which is never tight because it defines an open set.
For example, consider the familiar case of \Theorem{LLL}, and suppose that \eqref{eq:GLL} holds
with equality, i.e., $\Pr_\mu[E_i] = x_i \prod_{j \in \Gamma(i)} (1-x_j)$ for all $i$.
We show that the conclusion of the LLL remains true even if each event $E_i$ actually had the 
larger probability $\Pr_\mu[E_i] \cdot \big(1 + (2 \sum_i \frac{x_i}{1-x_i})^{-1}\big)$.
The proof of this fact crucially uses Shearer's criterion and it does not seem to follow from 
more elementary tools \cite{ErdosLovasz,Spencer77}.

\medskip
\noindent
{\bf Follow-up work.} Subsequently, Achlioptas and Iliopoulos generalized their framework further to
incorporate our notion of resampling oracles \cite{AchlioptasI16}.
This subsequent work can be viewed as a
unification of their framework and ours; it has the benefit of both capturing the framework of
resampling oracles and allowing some additional flexibility (in particular, the possibility of
regenerating the measure $\mu$ approximately rather than exactly). We remark that this work is still
incomparable with ours, primarily due to the facts that our analysis is performed in Shearer's more
general setting, and that our algorithm is efficient even when the LLL criteria are tight.

\medskip
\noindent
{\bf Organization.} 
The rest of the paper is organized as follows. In \Section{resample-existence}, we discuss the
connection between resampling oracles and the assumptions of the Lov\'asz Local Lemma. We also show
here that resampling oracles as well as the LLL itself can be computationally hard in general. In
\Section{implementation}, we show concrete examples of efficient implementations of resampling
oracles.
In \Section{applications} we discuss several applications of these resampling oracles.
Finally, in \Section{analysis} we present the full analysis of our algorithm.

\section{Resampling oracles: existence and efficiency}
\SectionName{resample-existence}

The algorithms in this paper make no reference to the lopsidependency condition
\eqref{eq:LLLL} and instead assume the existence of \ros.
In \Section{existence} we show that there is a close relationship between these two assumptions:
the existence of a \ro for each event is equivalent to the condition \Equation{Neg},
which is a strengthening of \eqref{eq:LLLL}.

We should emphasize that the {\em efficiency of an implementation} of a \ro is a separate issue.
There is no general guarantee that \ros can be implemented efficiently.
Indeed, as we show in \Section{hardness},
there are applications of the LLL such that the resampling oracles are hard to implement efficiently,
and finding a state avoiding all events is computationally hard,
under standard computational complexity assumptions.

Nevertheless, this is not an issue in common applications of the LLL: \ros exist and can be
implemented efficiently in all uses of the LLL of which we are aware,
even those involving lopsidependency.
\Section{implementation} has a detailed discussion of several scenarios.

\subsection{Existence of resampling oracles}
\SectionName{existence}

This section proves an equivalence lemma connecting \ros with the notion of \NAP.
First, let us define formally what we call a \ro.

\begin{definition}
\DefinitionName{resampling-oracle}
Let $E_1,\ldots,E_n$ be events on a space $\Omega$ with a probability measure $\mu$, and let $G = ([n], E)$ be a graph with neighbors of $i \in [n]$ denoted by $\Gamma(i)$. Let $r_i$ be a randomized procedure that takes a state $\omega \in \Omega$ and outputs a state $r_i(\omega) \in \Omega$.
We say that $r_i$ is a resampling oracle for $E_i$ with respect to $G$, if
\begin{description}
\item[\rm \RONE] For $\omega \sim \mu|_{E_i}$, we obtain $r_i(\omega) \sim \mu$.
(The oracle $r_i$ removes conditioning on $E_i$.)
\item[\rm \RTWO] For any $j \notin \Gamma^+(i) = \Gamma(i) \cup \{i\}$, if $\omega \not\in E_j$ then also $r_i(\omega) \not\in E_j$.
(Resampling an event cannot cause new non-neighbor events to occur.)
\end{description}
\end{definition}

Next, let us define the notion of a \NAG.
We denote by $E_i[\omega]$ the $\set{0,1}$-valued function
indicating whether $E_i$ occurs at a state $\omega \in \Omega$.

\newcommand{\myAnd}{\,\cap\,}

\begin{definition}
A graph $G$ with neighborhood function $\Gamma$ is a \NAG for events $E_1,\ldots,E_n$ if
\begin{equation}
\tag{LopA}
\EquationName{Neg}
\Pr_\mu[E_i \cap F] ~\geq~ \Pr_\mu[E_i] \cdot \Pr_\mu[F]
    \qquad \forall i \in [n], \forall F \in \cF_i
\end{equation}
where $\cF_i$ contains all events $F$ 
such that $F[\omega]$ is a monotone non-decreasing function of the functions
$(\, E_j[\omega] \,:\, j \notin \Gamma^+(i) \,)$.
\end{definition}

\begin{lemma}
\LemmaName{resample-existence}
Consider a fixed $i \in [n]$ and assume $\Pr_\mu[E_i] > 0$.
The following statements are equivalent.
\begin{description}
\item[(a)] There exists a \ro $r_i$ satisfying the conditions \RONE\ and \RTWO\ with respect to a neighborhood $\Gamma^+(i)$
\, (ignoring issues of computational efficiency).
\item[(b)] $\Pr_\mu[E_i \myAnd F] \geq \Pr_\mu[E_i] \cdot \Pr_\mu[F] $
for any event $F \in \cF_i$.
\end{description}
\end{lemma}

\begin{corollary}
Resampling oracles $r_1,\ldots,r_n$ exist for events $E_1,\ldots,E_n$ with respect to a graph $G$
if and only if $G$ is a \NAG for $E_1,\ldots,E_n$. 
Both statements imply that the lopsidependency condition \eqref{eq:LLLL} holds.
\end{corollary}

\begin{proofof}{\Lemma{resample-existence}}
(a) $\Rightarrow$ (b):
Consider the coupled states $(\omega,\omega')$
where $\omega \sim \mu|_{E_i}$ and $\omega' = r_i(\omega)$.
By \RONE, $\omega' \sim \mu$.
For any event $F \in \cF_i$, if $F$ does not occur at $\omega$ then it does not occur at $\omega'$
either, due to \RTWO.
This establishes that
$$
\Pr_\mu[F]
 ~=~ \E_{\omega' \sim \mu}[F[\omega']]
 ~\leq~ \E_{\omega \sim \mu|E_i}[F[\omega]]
 ~=~ \Pr_\mu[F \mid E_i],
$$
which implies $\Pr_\mu[F \myAnd E_i] \geq \Pr_\mu[F] \cdot \Pr_\mu[E_i]$.
In particular this implies \eqref{eq:LLLL}, by taking $F = \bigcup_{j \in J} {E_j}$.

(b) $\Rightarrow$ (a):
We begin by formulating the existence of a \ro as the following {\em transportation problem}.
Consider a bipartite graph $(U \cup W, E)$, where $U$ and $W$ are disjoint,
$U$ represents all the states $\omega \in \Omega$
satisfying $E_i$, and $W$ represents all the states $\omega \in \Omega$. Edges represent the
possible actions of the \ro: $(u,w) \in E$ if $u$ satisfies every event among
$(\, E_j \,:\, j \notin \Gamma^+(i) \,)$ that $w$ satisfies.
Each vertex has an associated  weight: For $w
\in W$, we define $p_w = \Pr_\mu[w]$, and for $u \in U$, $p_u = \Pr_\mu[u] / \Pr_\mu[E_i]$, i.e,
$p_u$ is the probability of $u$ conditioned on $E_i$.
We claim that the \ro $r_i$ exists if and only if there is an assignment $f_{uw}$ of values to the edges such that
\begin{equation}
\EquationName{transportation}
\begin{array}{lll}
&\sum_{w: (u,w) \in E} f_{uw} = p_u
    &\forall u \in U \\
&\sum_{u: (u,w) \in E} f_{uw} = p_w
    &\forall w \in W \\
&f_{uw} \geq 0 &\forall u \in U ,\, w \in W.
\end{array}
\end{equation}
Such an assignment is called a feasible transportation.
Given such a transportation, the \ro is defined naturally by following each edge from $u \in U$ with
probability $f_{uw} / p_u$, and the resulting distribution on $W$ is $p_w$. Conversely, for a \ro which,
for a given state $u \in U$, generates $w \in W$ with probability $q_{uw}$, we define $f_{uw} = p_u
q_{uw}$. This assignment satisfies \eqref{eq:transportation}.

Our goal at this point is show that (b) implies feasibility of \eqref{eq:transportation}.  
A condition that is equivalent to \eqref{eq:transportation}, but more convenient for our purposes,
can be determined from LP duality \cite[Theorem 21.11]{Schrijver}.
A feasible transportation exists if and only if
\begin{equation}
\EquationName{transportationdual}
\begin{array}{lll}
(\ref{eq:transportationdual}.1)& \sum_{u \in U} p_u ~=~ \sum_{w \in W} p_w \\
(\ref{eq:transportationdual}.2)& \sum_{u \in A} p_u ~\leq~ \sum_{w \in \Gamma(A)} p_w 
\qquad\forall A \subseteq U,
\end{array}
\end{equation}
where $\Gamma(A) = \setst{ w \in W }{ \exists u \in A \text{~s.t.~} (u,w) \in E }$.
This is an extension of Hall's condition for the existence of a perfect matching.

\newcommand{\simpledual}{(\ref{eq:transportationdual}$^*$)\xspace}

Our goal at this point is show that (b) implies feasibility of \eqref{eq:transportationdual}.
Let us now simplify \eqref{eq:transportationdual}.
Fix any $A \subseteq U$.
The neighborhood $\Gamma(A)$ consists of states satisfying at most those events among
$\setst{ E_j }{ j \notin \Gamma^+(i) }$ satisfied by some state in $A$.
Thus $\Gamma(A)$ corresponds to an event $F'$ such that $F'[\omega]$ is a 
\emph{non-increasing} function of $(\, E_j[\omega]: j \notin \Gamma^+(i) \,)$.
Next observe that, if the set of events among $\setst{ E_j }{ j \in \Gamma^+(i) }$ satisfied by
$u' \in U$ is a subset of those satisfied by $u \in U$, then $\Gamma(u') \subseteq \Gamma(u)$.
Suppose that, for each $u \in A$, we add to $A$ all such vertices $u'$.
Doing so can only increase the left-hand side of (\ref{eq:transportationdual}.2),
but does not increase the right-hand side as $\Gamma(A)$ remains unchanged
(since $\Gamma(u') \subseteq \Gamma(u)$).
Furthermore, the resulting set $A$ corresponds to the same event $F'$,
but restricted to the states in $U$.
Let us call such a set $A$ non-increasing.
Let \simpledual denote the simplification of \eqref{eq:transportationdual}
in which we restrict to non-increasing $A$.
We have argued that \eqref{eq:transportationdual} and \simpledual are equivalent.

Our goal at this point is show that (b) implies feasibility of \simpledual.
One may easily see that (b) is equivalent to
$$
\Pr_\mu[\overline{F} \myAnd E_i]
    ~\leq~
    \Pr_\mu[\overline{F}] \cdot \Pr_\mu[E_i]
    \qquad\forall F \in \cF_i.
$$
Assuming $\Pr[E_i] > 0$, we can rewrite this as
$\Pr_\mu[\overline{F} \mid E_i] \leq \Pr_\mu[\overline{F}] ~\forall F \in \cF_i$. 
Now consider using this inequality with $F = \overline{F'}$ for each $F'$ corresponding to 
some non-increasing set $A \subseteq U$.
We then have $\Pr_\mu[{F'} \mid E_i] = \sum_{u \in A} p_u$ and
$\Pr_\mu[{F'}] = \sum_{w \in \Gamma(A)} p_w$.
This verifies the feasibility of \simpledual.
\end{proofof}

\subsubsection{Example: monotone events on lattices}
\SectionName{resample-monotone}

This section presents an example of a setting where \Lemma{resample-existence} implies the existence
of a non-trivial \ro, even though the \NAG is empty.
This setting was previously known to have connections to the existential LLL \cite{LuMohrSzekely}.
The probability space here is $\Omega = \set{0,1}^M$,
viewed in the natural way as the Boolean lattice with operations $\wedge$ (meet) and $\vee$ (join),
and with the partial order denoted $\geq$.
Let $\mu : \set{0,1}^M \rightarrow [0,1]$ be a probability distribution,
i.e., $\sum_{x \in \set{0,1}^M} \mu(x) = 1$.
We assume that $\mu$ is \newterm{log-supermodular}, meaning that
$$
\mu(x \vee y) \mu(x \wedge y) ~\geq~ \mu(x) \mu(y)
\qquad \forall x,y \in \{0,1\}^M.
$$
As an example, any product distribution is log-supermodular.
Consider monotone increasing events $E_i$, i.e., such that $x' \geq x \in E_i  \Rightarrow x' \in E_i$. 
Note that any monotone increasing function of such events is again monotone increasing. 
It follows directly from the FKG inequality \cite{AlonSpencer} that condition (b) of
\Lemma{resample-existence} is satisfied for such events with an {\em empty} \NAG.
Therefore, a \ro exists in this setting.
However, the explicit description of its operation might be complicated and we do not know whether it can be implemented efficiently in general.

Alternatively, the existence of the \ro can be proved directly, using a theorem of Holley~\cite[Theorem 6]{Holley}. The \ro is described in Algorithm~\ref{alg:monotone}.
The reader can verify that this satisfies the assumptions \RONE\ and \RTWO, using Holley's Theorem.

\begin{algorithm}
\caption{\Ro for a monotone increasing event $E$. Let $\nu$ be the function guaranteed by
\Theorem{holley} when $\mu_1(x) = \frac{\mu(x) \b1_{x \in E}}{\sum_{e \in E} \mu(e)}$,
$\mu_2(y) = \mu(y)$, and $\b1_{x \in E}$ is the indicator function of $x \in E$.}
\label{alg:monotone}
\begin{algorithmic}[1]
\STATE {\bf Function} $r_E(x)$:
\STATE If $x \not\in E$, {\bf fail}.
\STATE Randomly select $y$ with probability $\frac{\nu(x,y)}{\sum_{y'} \nu(x,y')}$.
\RETURN $y$.
\end{algorithmic}
\end{algorithm}

\begin{theorem}[Holley's Theorem]
\TheoremName{holley}
Let $\mu_1$ and $\mu_2$ be probability measures on $\{0,1\}^{M}$ satisfying
$$
\mu_1(x \vee y) \mu_2(x \wedge y) ~\geq~ \mu_1(x) \mu_2(y)
\qquad\forall x, y \in \set{0,1}^M.
$$
Then there exists a probability distribution
$\nu : \set{0,1}^{M} \times \set{0,1}^{M} \rightarrow \RR$ satisfying
\begin{gather*}
\mu_1(x) = \smallsum{y}{} \, \nu(x,y) \\
\mu_2(y) = \smallsum{x}{} \, \nu(x,y) \\
\nu(x,y)=0 ~\:\text{unless}~\: x \geq y.
\end{gather*}
\end{theorem}

\subsection{Computational hardness of the LLL}
\SectionName{hardness}

This section considers whether the LLL can always be made algorithmic.
We show that, even in fairly simple scenarios where the LLL applies,
finding the desired output can be computationally hard,
a fact that surprisingly seems to have been overlooked.
We first observe that the question of algorithmic efficiency must be stated carefully otherwise
hardness is trivial.

\vspace{4pt}
\noindent
\textit{A trivial example.}
Given a Boolean formula $\phi$, let the probability space be $\Omega = \set{0,1}$,
and let $\mu$ be the uniform measure on $\Omega$.
There is a single event $E_1$ defined to be $E_1 = \set{1}$ if $\phi$ is satisfiable,
and $E_1 = \set{0}$ if $\phi$ is not satisfiable.
Since $\Pr[E_1]=1/2$, the \eqref{eq:GLL} criterion holds trivially with $x_1 = 1/2$.
The LLL gives the obvious conclusion that there is a state $\omega \notin E$.
Yet, finding this state requires deciding satisfiability of $\phi$, which is NP-complete.
\vspace{4pt}

The reason that this example is trivial is that even deciding whether the undesired
event has occurred is computationally hard.
A more meaningful discussion of LLL efficiency ought to rule out this trivial example
by considering only scenarios that satisfy some reasonable assumptions.
With that in mind, we will assume that
\begin{compactitem}
\item there is a probability space $\Omega$, whose states can be described by $m$ bits;
\item a graph $G$ satisfying \eqref{eq:Dep} for events $E_1,\ldots,E_n$ is explicitly provided;
\item $x_1,\ldots,x_n \in (0,1)$ satisfying the \eqref{eq:GLL} conditions are provided,
and $\sum_{i=1}^{n} \frac{x_i}{1-x_i}$ is at most $\poly(n)$;
\item there is a subroutine that provides an independent random state
      $\omega \sim \mu$ in $\poly(m)$ time;
\item for each $i \in [n]$, there is a subroutine which determines for any given $\omega \in \Omega$ whether $\omega \in E_i$, in $\poly(m)$ time.
\end{compactitem}
As far as we know, no prior work refutes the possibility that there is an algorithmic
form of the LLL, with running time $\poly(m,n)$, in this general scenario.

Our results imply that \ros do \emph{exist} in this general scenario, so
it is only the question of whether these \ros are \emph{efficient} that prevents
\Theorem{LLL-tight-result} from providing an efficient algorithm.
Nevertheless, we show that there is an instance of the LLL that satisfies the reasonable assumptions
stated above, but for which finding a state in $\bigcap_i \overline{E_i}$ requires solving
a problem that is computationally hard (under standard computational complexity assumptions).
As a consequence, we conclude that the \ros cannot always be implemented efficiently,
even under the reasonable assumptions of this general scenario.

\newcommand{\GF}{\mathrm{GF}}

We remark that NP-completeness is not the right notion of hardness here \cite{Papadimitriou}.
Problems in NP involve deciding whether a solution exists, whereas the LLL \emph{guarantees that a
solution exists}, and the goal is to explicitly find a solution.
Our result is instead based on hardness of the \emph{discrete logarithm} problem,
a standard belief in computational complexity theory.
In the following, $\GF(p^n)$ for a prime $p$ and integer $n$ denotes a finite field of order $p^n$,
and $\GF^*(p^n)$ its multiplicative group of nonzero elements.

\begin{theorem}
\TheoremName{LLL-hardness}
There are instances of events $E_1,\ldots,E_n$ on a probability space $\Omega = \{0,1\}^n$ under the uniform probability measure, such that
\begin{compactitem}
\item the events $E_i$ are mutually independent;
\item for each $i \in [n]$, the condition $\omega \in E_i$ can be checked in $\poly(n)$ time for
given $\omega \in \Omega$;
\item the \eqref{eq:GLL} conditions are satisfied with $x_i = 1/2$ for each $i \in [n]$;
\end{compactitem}
but finding a state in $\bigcap_{i=1}^{n} \overline{E_i}$ is as hard as solving the
discrete logarithm problem in $\GF^*(2^n)$.
\end{theorem}

\noindent
\textit{Remark.} Superficially, this result seems to contradict
the fact that the LLL can be made algorithmic in the variable model \cite{MoserTardos},
where events are defined on underlying independent random variables.
The key point is that the variable model also relies on a particular type of dependency graph
(defined by shared variables) which might be more conservative than the true dependencies between
the events.
\Theorem{LLL-hardness} shows that, even if the probability space consists of independent $\{0,1\}$ random variables, the LLL cannot in general be made algorithmic if the true dependencies are considered.

\begin{proof}
Consider an instance of the discrete logarithm problem in the multiplicative group $\GF^*(2^n)$. The
input is a generator $g$ of $\GF^*(2^n)$ and an element $h \in \GF^*(2^n)$. The goal is to find an integer $1 \leq k \leq 2^n-1$ such that $g^k = h$. We define an instance of $n$ events on $\Omega = \{0,1\}^n$ as follows.

We identify $\Omega = \{0,1\}^n$ with $[2^n]$ as well as $\GF(2^n)$ in a natural way. 
We define $f:[2^n] \rightarrow \GF(2^n)$ by $f(0) = 0$ and $f(x) = g^x$ for $x \neq 0$, where the
exponentiation is performed in $\GF(2^n)$. 
For each $i \in [n]$, we define an event $E_i$ that occurs for $\omega \in \{0,1\}^n$ iff $(f(\omega))_i = 1-h_i$. This is a condition that can be checked in time $\poly(n)$, by computing $f(\omega) = g^\omega$ where we interpret $\omega$ as $\sum_{i=0}^{n-1} \omega_i 2^i$ and compute $g^\omega$ by taking squares iteratively.

Observe that for $\omega$ distributed uniformly in $\Omega = \{0,1\}^n$, $f(\omega)$ is again
distributed uniformly in $\Omega$, since $f$ is a bijection ($0$ is mapped to $0$, and $f(\omega)$
for $\omega \neq 0$ generates each element of the multiplicative group $\GF^*(2^n)$ exactly once).
Therefore, the probability of $E_i$ is $1/2$, for each $i \in [n]$. Further, the events
$E_1,\ldots,E_n$ are mutually independent, since for any $J \subseteq [n]$, $\bigcap_{j \in J}
{E_j} \cap \bigcap_{j' \notin J} {\overline{E_{j'}}}$ occurs iff $f(\omega) = h \oplus \b1_J$, which
happens with probability $1/2^n$.
Here $\b1_J \in \set{0,1}^n$ is the indicator vector for the set $J$,
and $\oplus$ denotes addition in $\GF(2^n)$
(i.e., component-wise xor in $\set{0,1}^n$).
Hence the dependency graph is empty, and the LLL with
parameters $x_i = 1/2$ trivially implies that there exists a state $\omega$ avoiding all the events.
In this instance, we know explicitly that the state avoiding all the events is $f^{-1}(h)$. Therefore, if we had an efficient algorithm to find this point for any given $h \in \GF^*(2^n)$, we would also have an efficient algorithm for the discrete logarithm problem in $\GF(2^n)$. 
\end{proof}

\section{Implementation of resampling in specific settings}
\SectionName{implementation}

In this section, we present efficient implementations of \ros in four application settings:
independent random variables (which was the setting of \cite{MoserTardos}),
random permutations (handled by \cite{HarrisS14}),
perfect matchings in complete graphs (some of whose applications are made algorithmic
by \cite{Achlioptas}),
and spanning trees in complete graphs (which is a new scenario that we can handle).
To be more precise, resampling oracles also depend on the types of events and dependencies that we want to handle.\footnote{In \Section{hardness} we give an example of events on independent random variables for which resampling oracles exist but cannot be made efficient.} In the setting of independent random variables, we can handle arbitrary events with dependencies defined by overlapping relevant variables, just like \cite{MoserTardos}.
In the setting of permutations, we handle the appearance of patterns in permutations as in \cite{HarrisS14}.
In the settings of matchings and spanning trees, we consider the ``canonical events" defined by
\cite{LuMohrSzekely}, characterized by the appearance of a certain subset of edges. We also show in
\Section{product-resampling} how resampling oracles for a certain probability space can be extended
in a natural way to products of such probability spaces (for example, how to go from resampling oracles for one random permutation to a collection of independent random permutations).
These settings cover all the applications of the lopsided LLL that we are aware of.

\subsection{The variable model}
\SectionName{resample-independent}

This is the most common setting, considered originally by Moser and Tardos \cite{MoserTardos}. Here,
$\Omega$ has a product structure corresponding to independent random variables $\setst{ X_a }{ a \in
\cU }$. The probability measure $\mu$ here is a product measure. Each bad event $E_i$ depends on a particular subset of variables $A_i$, and two events are independent iff $A_i \cap A_j = \emptyset$. 

Here our algorithmic assumptions correspond exactly to the Moser-Tardos framework
\cite{MoserTardos}. Sampling from $\mu$ means generating a fresh set of random variables
independently. The \ro $r_i$ takes a state $\omega$ and replaces the random variables $\setst{ X_a
}{ a \in A_i }$ by fresh random samples. It is easy to see that the assumptions are satisfied: in
particular, a random state sampled from $\mu$ conditioned on $E_i$ has all variables outside of
$A_i$ independently random. Hence, resampling the variables of $A_i$ produces the distribution
$\mu$. Clearly, resampling $\setst{X_a }{ a \in A_i }$ does not affect any events 
whose variables do not intersect $A_i$.

We note that this \ro is also consistent with the notion of lopsidependency on product spaces considered by \cite{MoserTardos}: They call two events $E_i, E_j$ lopsidependent, if $A_i \cap A_j \neq \emptyset$ and it is possible to cause $E_j$ to occur by resampling $A_i$ in a state where $E_i$ holds but $E_j$ does not (the definition in \cite{MoserTardos} is worded differently but equivalent to this). This is exactly the condition that we require our \ro to satisfy.

\subsection{Permutations}
\SectionName{resample-permutations}

\newcommand{\vbl}{\operatorname{vbl}}

The probability space $\Omega$ here is the space of all permutations $\pi$ on a set $[n]$, with a uniform measure $\mu$.
The bad events are assumed to be ``simple" in the following sense: Each bad event $E_i$ is defined by a ``pattern"
$P(E_i) = \{ (x_1,y_1), \ldots, (x_{t(i)}, y_{t(i)})\}$. The event $E_i$
occurs if $\pi(x_j) = y_j$ for each $1 \leq j \leq t(i)$. Let $\vbl(E_i) = \setst{x}{ \exists y,
(x,y) \in P(E_i) }$ denote the variables of $\pi$ relevant to event $E_j$. Let us define a
relation $i \sim i'$ to hold iff there are pairs $(x,y) \in P(E_i), (x',y') \in P(E_{i'})$
such that $x=x'$ or $y=y'$; i.e., the two events entail the same value in either the range or domain.
This relation defines a \NDG. It is known that the lopsided LLL holds in this setting. 

\begin{algorithm}
\caption{\Ro for permutations}
\label{alg:permutation-shuffle}
\begin{algorithmic}[1]
\STATE {\bf Function} $r_i(\pi$):
\STATE $X := \vbl(E_i)$, i.e., the variables in $\pi$ affecting event $E_i$;
\STATE Fix an arbitrary order $X = (x_1,x_2,\ldots,x_t)$;
\FOR {$i=t$ down to $1$}
\STATE Swap $\pi(x_i)$ with $\pi(z)$ for $z$ uniformly random among $[n] \setminus \{x_1,\ldots,x_{i-1}\}$;
\ENDFOR
\RETURN $\pi$;
\end{algorithmic}
\end{algorithm}

Harris and Srinivasan \cite{HarrisS14} showed how, under the LLL criteria, a permutation avoiding all bad events can be found algorithmically. We implement the \ro based on their algorithm (see Algorithm~\ref{alg:permutation-shuffle}). 
To prove the correctness of this \ro within our framework, we need the following lemma.

\begin{lemma}
\label{lem:shuffle}
Suppose that a permutation $\pi$ has some arbitrary fixed assignment on the variables in $X$, $\pi|_{X} = \phi$, and it is uniformly random among all permutations satisfying $\pi|_{X} = \phi$. Then the output of Shuffle$(\pi,X)$ is a uniformly random permutation.
\end{lemma}

The procedure is known as the Fisher-Yates shuffle for generating uniformly random permutations (and was used in \cite{HarrisS14} as well). In contrast to the full shuffle, we assume that some part of the permutation has been shuffled already: $X$ is the remaining portion that still remains to be shuffled, and conditioned on its assignment the rest is uniformly random. This would be exactly the distribution achieved after performing the Fisher-Yates shuffle on the complement of $X$. Our procedure performs the rest of the Fisher-Yates shuffle, which produces a uniformly random permutation. For completeness we give a self-contained proof.

\begin{proof}
Let $X = \{x_1,\ldots,x_t\}$.
By induction, after performing the swap for $x_i$, the permutation is uniform among all permutations with a fixed assignment of $\{x_1,\ldots,x_{i-1}\}$ (consistent with $\phi$). This holds because, before the swap, the permutation was by induction uniform conditioned on the assignment of $\{x_1,\ldots,x_i\}$ being consistent with $\phi$, and we choose a uniformly random swap for $x_i$ among the available choices. This makes every permutation consistent with $\phi$ on $\{x_1,\ldots,x_{i-1}\}$ equally likely after this swap.
\end{proof}

This verifies the first condition for our \ro.
The second condition is that resampling of {\em occurring events} does not affect non-neighbor events. This is true because of the following lemma.

\begin{lemma}
\label{lem:dependency}
The \ro $r_i(\pi)$ applied to a permutation satisfying $E_i$ does not cause any new event outside of $\Gamma^+(I)$ to occur.
\end{lemma}

\begin{proof}
Suppose $E_j$ changed its status during a call to $r_i(\pi)$. This means that something changed
among its relevant variables $\vbl(E_j)$. This could happen in two ways:

 (1) either a variable $z \in \vbl(E_j)$ was swapped because $z \in X = \vbl(E_i)$; then clearly $j \in \Gamma^+(i)$.

 (2) or, a variable in $\vbl(E_j)$, although outside of $X$,  received a new value by a swap with
 some variable in $X = \vbl(E_i)$. Note that in the Shuffle procedure, every time a variable $z$ outside of $X$ changes its value, it is by a swap with a fresh variable of $X$, i.e. one that had not been processed before. Therefore, the value that $z$ receives is one that previously caused $E_i$ to occur. If it causes $E_j$ to occur, it means that $E_i$ and $E_j$ share a value in the range space and we have $j \in \Gamma^+(i)$ as well.
\end{proof}

\subsection{Perfect matchings}
\SectionName{resample-matchings}

Here, the probability space $\Omega$ is the set of all perfect matchings in $K_{2n}$, with the uniform measure. 
This is a setting considered by \cite{Achlioptas} and it is also related to the setting of
permutations. (Permutations on $[n]$ can be viewed as perfect matchings in $K_{n,n}$.) A state here
is a perfect matching in $K_{2n}$, which we denote by $M \in \Omega$. We consider bad events of the
following form: $E_A$ for a set of edges $A$ occurs if $A \subseteq M$. Obviously, $\Pr_\mu[E_A] >
0$ only if $A$ is a (partial) matching. Let us define $A \sim B$ iff $A \cup B$ is {\em not} a
matching. It was proved in \cite{LuMohrSzekely} that this defines a \NDG. 

Our goal is to implement a \ro in this setting. We describe such an operation in Algorithm~\ref{alg:matching-shuffle}.

\begin{algorithm}
\caption{\Ro for perfect matchings}
\label{alg:matching-shuffle}
\begin{algorithmic}[1]
\STATE {\bf Function} $r_{A}(M)$:
\STATE Check that $A \subseteq M$, otherwise {\bf return} $M$.
\STATE $A' := A$;
\STATE $M' := M$;
\WHILE {$A' \neq \emptyset$}
\STATE Pick $(u,v) \in A'$ arbitrarily;
\STATE Pick $(x,y) \in M' \setminus A'$ uniformly at random, with $(x,y)$ randomly ordered;
\STATE With probability $1 - \frac{1}{2|M' \setminus A'|+1}$,
\STATE \hspace{12pt} Add $(u,y), (v,x)$ to $M'$ and remove $(u,v),(x,y)$ from $M'$;
\STATE Remove $(u,v)$ from $A'$;
\ENDWHILE
\RETURN $M'$.
\end{algorithmic}
\end{algorithm}

\begin{lemma}
Let $A$ be a matching in $K_{2n}$ and let $M$ be distributed uniformly among perfect matchings in $K_{2n}$ such that $A \subseteq M$.
Then after calling the \ro, $r_A(M)$ is a uniformly random perfect matching.
\end{lemma}

\begin{proof}
We prove by induction that at any point, $M'$ is a uniformly random perfect matching conditioned on containing $A'$.
This is satisfied at the beginning: $M'=M, A'=A$ and $M$ is uniformly random conditioned on $A \subseteq M$.

Assume this is true at some point, we pick $(u,v) \in A'$ arbitrarily and $(x,y) \in M' \setminus A'$ uniformly at random.
Denote the vertices covered by $M' \setminus A'$ by $V(M' \setminus A')$. 
Observe that for a uniformly random perfect matching on $V(M' \setminus A') \cup \{u,v\}$, the edge $(u,v)$ should appear with probability $1 / (2|M' \setminus A'|+1)$ since $u$ has $2|M' \setminus A'|+1$ choices to be matched with and $v$ is 1 of them. Consequently, we keep the edge $(u,v)$ with probability $1 / (2|M' \setminus A'|+1)$ and conditioned on this $M' \setminus A'$ is uniformly random by the inductive hypothesis. Conditioned on $(u,v)$ not being part of the matching, we re-match $(u,v)$ with another random edge $(x,y) \in M' \setminus A'$ where $(x,y)$ is randomly ordered. In this case, $u$ and $v$ get matched to a uniformly random pair of vertices $x,y \in V(M' \setminus A')$, as they should be. The rest of the matching $M' \setminus A' \setminus \{(x,y)\}$ is uniformly random on $V(M' \setminus A' \setminus \{x,y\})$ by the inductive hypothesis. 

Therefore, after each step $M' \setminus A'$ is uniformly random conditioned on containing $A'$. At the end, $A'=\emptyset$ and $M'$ is uniformly random. 
\end{proof}

\begin{lemma}
The \ro $r_A(M)$ applied to a perfect matching satisfying event $E_A$ does not cause any new event $E_B$ such that $B \notin \Gamma^+(A)$.
\end{lemma}

\begin{proof}
Observe that all the new edges that the \ro adds to $M$ are incident to some vertex matched by $A$. So if an event $E_B$ was not satisfied before the operation and it is satisfied afterwards, it must be the case that $B$ contains some edge not present in $A$ but sharing a vertex with $A$. Hence, $A \cup B$ is not a matching and $A \sim B$.
\end{proof}

\subsection{Spanning trees}
\SectionName{resample-trees}

Here, the probability space $\Omega$ is the set of all spanning trees in $K_n$.
Let us consider events $E_A$ for a set of edges $A$, where $E_A$ occurs for $T \in \Omega$ iff $A
\subseteq T$. Define $A \sim B$ for distinct $A,B$ unless $A$ and $B$ are vertex-disjoint.
Lu et al.~\cite[Lemma 7]{LuMohrSzekely} show that this in fact defines a \emph{dependency}
graph for spanning trees.
It is worth emphasizing that in this scenario the \eqref{eq:Dep} condition holds
(the more general condition \eqref{eq:LLLL} is not needed),
but the scenario does not fall within the scope of the Moser-Tardos variable model.
It does fall within the scope of our framework, but one must design a non-trivial \ro.

To implement a \ro in this setting, we will use as a subroutine an algorithm to generate a uniformly
random spanning tree in a given graph $G$. This can be done efficiently by several methods,
for example by a random walk \cite{Broder}.

\begin{algorithm}
\caption{\Ro for spanning trees}
\label{alg:cycle-shuffle}
\begin{algorithmic}[1]
\STATE {\bf Function} $r_{A}(T)$:
\STATE Check that $A \subseteq T$, otherwise {\bf fail}.
\STATE Let $W = V(A)$, the vertices covered by  $A$.
\STATE Let $T_1 = {V \setminus W \choose 2} \cap T$, the edges of $T$ disjoint from $W$.
\STATE Let $F_1 = {V \setminus W \choose 2} \setminus T$, the edges disjoint from $W$ not present in $T$.
\STATE Let $G_2 = (K_n \setminus F_1) / T_1$ be a multigraph obtained by deleting $F_1$ and contracting $T_1$.
\STATE Generate a uniformly random spanning tree $T_2$ in $G_2$.
\RETURN $T_1 \cup T_2$.
\end{algorithmic}
\end{algorithm}

\begin{lemma}
\label{lem:tree-decondition}
If $A$ is a fixed forest and $T$ is a uniformly random spanning tree in $K_n$ conditioned on $A \subseteq T$, then $r_A(T)$ produces a uniformly random spanning tree in $K_n$.
\end{lemma}

\begin{proof}
First, observe that since $T_2$ is a spanning tree of $G_2 = (K_n  \setminus F_1) / T_1$, it is also a spanning tree of $K_n / T_1$ where $T_1$ is a forest, and therefore $T_1 \cup T_2$ is a spanning tree of $K_n$. We need to prove that it is a uniformly random spanning tree.

First, we appeal to a known result \cite[Lemma 6]{LuMohrSzekely} stating that given a forest $F$ in
$K_n$ with components of sizes (number of vertices) $f_1, f_2, \ldots, f_m$, the number of spanning trees containing $F$ is exactly 
\begin{equation}
\EquationName{countsptree}
 n^{n-2} \prod_{i=1}^{m} \frac{f_i}{n^{f_i-1}}.
\end{equation}
Equivalently (since $n^{n-2}$ is the total number of spanning trees), for a uniformly random spanning tree $T$, $\Pr[F \subseteq T] = \prod_{i=1}^{m} f_i/n^{f_i-1}$. This has the surprising consequence that for vertex-disjoint forests $F_1, F_2$, we have $\Pr[F_1 \cup F_2 \subseteq T] = \Pr[F_1 \subseteq T] \cdot \Pr[F_2 \subseteq T]$, i.e., the containment of $F_1$ and $F_2$ are independent events. (In a general graph, the appearances of different edges in a random spanning tree are negatively correlated, but here we are in a complete graph.) 

Let $W = V(A)$ and let $B$ be any forest on $V \setminus W$, i.e., vertex-disjoint from $A$. By the above, the appearance of $B$ in a uniformly random spanning tree is independent of the appearance of $A$. Hence, if $T$ is uniformly random, we have $\Pr[B \subseteq T \mid A \subseteq T] = \Pr[B \subseteq T]$. This implies that the distribution of $T \cap {V \setminus W \choose 2}$ is exactly the same for a uniformly random spanning tree $T$ as it is for one conditioned on $A \subseteq T$ (formally, by applying the inclusion-exclusion formula). Therefore, the forest $T_1 = T \cap {V \setminus W \choose 2}$ is distributed as it should be in a random spanning tree restricted to $V \setminus W$.

The final step is that we extend $T_1$ to a spanning tree $T_1 \cup T_2$, where $T_2$ is a uniform spanning tree in $G_2 = (K_n \setminus F_1) / T_1$. Note that $G_2$ is a multigraph, i.e.,~it is important that we preserve the multiplicity of edges after contraction. The spanning trees $T_2$ in $G_2 = (K_n \setminus F_1) / T_1$ are in a one-to-one correspondence with spanning trees in $K_n$ conditioned on $T \cap {V \setminus W \choose 2} = T_1$. This is because each such tree $T_2$ extends $T_1$ to a different spanning tree of $K_n$, and each spanning tree where $T \cap {V \setminus W \choose 2} = T_1$ can be obtained in this way. Therefore, for a fixed $T_1$, $T_1 \cup T_2$ is a uniformly random spanning tree conditioned on $T \cap {V \setminus W \choose 2} = T_1$. Finally, since the distribution of $T_1$ is equal to that of a uniformly random spanning tree restricted to $V \setminus W$, $T_1 \cup T_2$ is a uniformly random spanning tree.
\end{proof}

\begin{lemma}
The \ro $r_A(T)$ applied to a spanning tree satisfying $E_A$ does not cause any new event $E_B$ such that $B \notin \Gamma^+(A)$.
\end{lemma}

\begin{proof}
Note that the only edges that we modify are those incident to $W = V(A)$. Therefore, any new event $E_B$ that the operation of $r_A$ could cause must be such that $B$ contains an edge incident to $W$ and not contained in $A$. Such an edge shares exactly one vertex with some edge in $A$ and hence $B \sim A$. 
\end{proof}

\subsection{Composition of resampling oracles for product spaces}
\SectionName{product-resampling}

Suppose we have a product probability space $\Omega = \Omega_1 \times \Omega_2 \times \ldots \times \Omega_N$, where on each $\Omega_i$ we have resampling oracles $r_{ij}$ for events $E_{ij}, j \in \cE_i$, with respect to a graph $G_i$. Our goal is to show that there is a natural way to combine these resampling oracles in order to handle events on $\Omega$ that are obtained by taking intersections of the events $E_{ij}$. The following theorem formalizes this notion.

\begin{theorem}
\TheoremName{thm:product-resampling}
Let $\Omega_1,\ldots,\Omega_N$ be probability spaces, where for each $\Omega_i$ we have resampling oracles $r_{ij}$ for events $E_{ij}, j \in \cE_i$ with respect to a graph $G_i$. Let $\Omega =
\Omega_1 \times \Omega_2 \times \ldots \Omega_N$ be a product space with the respective product
probability measure. For any set $J$ of pairs $(i,j), j \in \cE_i$ where each $i \in [N]$ appears at
most once, define an event $E_J$ on $\Omega$ to occur in a state $\omega =
(\omega_1,\ldots,\omega_N)$ iff $E_{ij}$ occurs in $\omega_i$ for each $(i,j) \in J$. Define a graph
$G$ on these events by $J \sim J'$ iff there exist pairs $(i,j) \in J, (i,j') \in J'$ such that $j
\sim j'$ in $G_i$. Then there exist resampling oracles $r_J$ for the events $E_J$ with respect to
$G$, which are obtained by calling in succession each of the oracles $r_{ij}$ for $(i,j) \in J$.
\end{theorem}

\begin{proof}
For notational simplicity, let us assume that on each $\Omega_i$ we have a trivial event $E_{i0} = \Omega_i$ and the respective resampling oracle $r_{i0}$ is the identity on $\Omega_i$. Then we can assume that each collection of events $J$ is in the form $J = \{ (1,j_1), (2,j_2), \ldots, (N,j_N) \}$, where we set $j_\ell = 0$ for components where there is no event to resample. We define
 $$r_J(\omega_1,\ldots,\omega_N) = (r_{1 j_1}(\omega_1), r_{2 j_2}(\omega_2), \ldots, r_{N j_N}(\omega_N)). $$
We claim that these are resampling oracles with respect to $G$ as defined in the theorem.

Let us denote by $\mu_i$ the probability distribution on $\Omega_i$ and by $\mu$ the product distribution on $\Omega$.
For the first condition, suppose that $\omega \sim \mu|_{E_J}$. By the product structure of
$\Omega$, this is the same as having $\omega = (\omega_1, \ldots, \omega_N)$ where the components
are independent and $\omega_\ell \sim \mu_\ell|_{E_{\ell j_\ell}}$ for each $(\ell,j_\ell) \in J$,
and $\omega_\ell \sim \mu_\ell$ for components such that $j_\ell=0$. By the properties of the
resampling oracles $r_{\ell j_\ell}$, we have $r_{\ell j_\ell}(\omega_\ell) \sim \mu_\ell$. Since the resampling oracles are applied with independent randomness for each component, we have
$$
r_J(\omega)
 ~=~ (r_{1 j_1}(\omega_1), r_{2 j_2}(\omega_2), \ldots, r_{N j_N}(\omega_N))
 ~\sim~ \mu_1 \times \mu_2 \times \ldots \times \mu_N
 ~=~ \mu.
$$

For the second condition, note that if $\omega \notin E_{J'}$ and $r_J(\omega) \in E_{J'}$, it must be the case that there is $(\ell,j_\ell) \in J$ and $(\ell,j'_\ell) \in J'$ such that $\omega_\ell \notin E_{\ell j'_\ell}$ and $r_{\ell j_\ell}(\omega) \in E_{\ell j'_\ell}$. However, this is possible only if $j_\ell \sim j'_\ell$ in the graph $G_{\ell}$. By the definition of $G$, this means that $J \sim J'$ as well.
\end{proof}

As a result, we can extend our resampling oracles to spaces like $N$-tuples of independent random
permutations, independent random spanning trees, etc. Such extensions are used in our applications.

\section{Applications}
\SectionName{applications}

Let us present a few applications of our framework. Our application to rainbow spanning trees is new, even in the existential sense. Our applications to Latin transversals and rainbow matchings are also new to the best of our knowledge, although they could also have been obtained using the framework of \cite{HarrisS14} and \cite{Achlioptas}. 

\subsection{Rainbow spanning trees}

Given an edge-coloring of $K_n$, a spanning tree is called rainbow if each of its edges has a distinct color.
The existence of a single rainbow spanning tree is completely resolved by the matroid intersection theorem: It can be decided efficiently whether a rainbow spanning tree exists for a given edge coloring, and it can be found efficiently if it exists. However, the existence of multiple edge-disjoint rainbow spanning trees is more challenging.
An attractive conjecture of Brualdi and Hollingsworth \cite{Brualdi} states that if $n$ is even and $K_n$ is properly edge-colored by $n-1$ colors, then the edges can be decomposed into $n/2$ rainbow spanning trees, each tree using each color exactly once. Until recently, it was only known that every such edge-coloring contains $2$ edge-disjoint rainbow spanning trees \cite{Akbari}. In a recent development, it was proved that if every color is used at most $n/2$ times (which is true for any proper coloring) then there exist $\Omega(n / \log n)$ edge-disjoint rainbow spanning trees \cite{Carraher}. In fact this result seems to be algorithmically efficient, although this was not claimed by the authors. We prove that using our  framework, we can find $\Omega(n)$ rainbow spanning trees under a slight strengthening of the coloring assumption.

\begin{theorem}
\label{thm:rainbow-trees}
Given an edge-coloring of $K_n$ such that each color appears on at most $\frac{1}{32} (\frac78)^7 n$ edges, at least
$\frac{1}{32} (\frac78)^7 n$ edge-disjoint rainbow spanning trees exist and can be found in $O(n^4)$ \ro calls with high probability.
\end{theorem}

This result relies on \Theorem{cluster-no-slack}, our algorithmic version of the LLL under
the cluster expansion criterion. To obtain the result with high probability, we appeal to a more refined bound that we state in \Theorem{cluster-with-slack}. We note that if there is constant multiplicative slack in the assumption on color appearances, the number of resamplings improves to $O(n^2)$, using the result in \Theorem{cluster-with-slack} with constant $\epsilon$ slack.

To prove the existential statement, we simply sample $\frac{1}{32} (\frac78)^7 n$ independently
random spanning trees and hope that they will be (a) pairwise edge-disjoint, and (b) rainbow. This
unlikely proposition happens to be true with positive probability, thanks to the LLL and the independence properties of random spanning trees that we mentioned in Section~\ref{sec:resample-trees}. Given this setup, our framework implies that we can also find the rainbow trees efficiently.

\begin{proof}
We apply our algorithm in the setting of $t$ independent and uniformly random spanning trees $T_1,\ldots,T_t \subset K_{n}$, with the following two types of bad events:
\begin{itemize}
\item $E^i_{ef}$: For each $i \in [t]$ and two edges $e \neq f$ in $K_n$ of the same color, $E^i_{ef}$ occurs if $\{e,f\} \subset T_i$;
\item $E^{ij}_e$: For each $i \neq j \in [t]$ and an edge $e$ in $K_n$, $E^{ij}_e$ occurs if $e \in T_i \cap T_j$.
\end{itemize}
Clearly, if no bad event occurs then the $t$ trees are rainbow and pairwise edge-disjoint.

By \eqref{eq:countsptree} the probability of a bad event of the first type is $\Pr[E^i_{ef}] = 3/n^2$ if $|e \cup f| = 3$ and $\Pr[E^i_{ef}] =4/n^2$ if $|e \cup f| = 4$. The probability of a bad event of the second type is $\Pr[E^{ij}_e] = (2/n)^2 = 4/n^2$, since each of the two trees contains $e$ independently with probability $2/n$. Hence, the probability of each bad event is upper-bounded by $p = 4/n^2$.

In Section~\ref{sec:resample-trees} we constructed a \ro $r_A$ for a single spanning tree. By \Theorem{thm:product-resampling}, this \ro extends in a natural way to the setting of $t$ independent random spanning trees.
In particular, for an event $E^{i}_{ef}$, we define $r^i_{ef}$ as an application of the \ro $r_{\{e,f\}}$ to the tree $T_i$. For an event $E^{ij}_e$, we define $r^{ij}_e$ as an application of the \ro $r_{\{e\}}$ independently to the trees $T_i$ and $T_j$. It is easy to check using \Theorem{thm:product-resampling} that for independent uniformly random spanning trees conditioned on either type of event, the respective \ro generates independent uniformly random spanning trees.

Let us define the following dependency graph; we are somewhat conservative for the sake of simplicity. The graph contains the following kinds of edges:
\begin{itemize}
\item $E^{i}_{ef} \sim E^{i}_{e'f'}$ whenever $e \cup f$ intersects $e' \cup f'$;
\item $E^{i}_{ef},E^{j}_{ef} \sim E^{ij}_{e'}$ whenever $e'$ intersects $e \cup f$;
\item $E^{ij}_{e} \sim E^{ij'}_{e'}, E^{i'j}_{e'}$ whenever $e'$ intersects $e$.
\end{itemize}

We claim that the \ro for any bad event can cause new bad events only in its neighborhood. This
follows from the fact that the \ro affect only the trees relevant to the event (in the superscript), and the only edges modified are those incident to those relevant to the event (in the subscript).

Let us now verify the cluster expansion criterion, introduced as \eqref{eq:CLL}
in \Section{generalizingLLL}, so that we may apply \Theorem{cluster-with-slack}.
Let us assume that each color appears on at most $q$ edges, and we generate $t$ random spanning trees. We claim that the neighborhood of each bad event can be partitioned into $4$ cliques of size $(n-1)(t-1)$ and $4$ cliques of size $(n-1)(q-1)$.

First, let us consider an event of type $E^{i}_{ef}$. The neighborhood of $E^{i}_{ef}$ consists of: (1) events $E^{i}_{e'f'}$ where $e'$ or $f'$ shares a vertex with $e \cup f$; these events form $4$ cliques, one for each vertex of $e \cup f$, and the size of each clique is at most $(n-1)(q-1)$, since the number of incident edges to a vertex is $n-1$, and the number of other edges of the same color is at most $q-1$. (2) events $E^{ij}_{e'}$ where $e'$ intersects $e \cup f$; these events form $4$ cliques, one for each vertex of $e \cup f$, and each clique has size at most $(n-1)(t-1)$, since its events can be identified with the $(n-1)$ edges incident to a fixed vertex and the remaining $t-1$ trees.

Second, let us consider an event of type $E^{ij}_{e}$. The neighborhood of $E^{ij}_{e}$ consists of: (1) events $E^{i}_{e'f'}$ and $E^{j}_{e'f'}$ where $e$ intersects $e' \cup f'$; these events form $4$ cliques, one for each vertex of $e$ and either $i$ or $j$ in the superscript, and the size of each clique is at most $(n-1)(q-1)$ by an argument as above. (2) events $E^{i'j}_{e'}, E^{ij'}_{e'}$ where $e'$ intersects $e$; these events form $4$ cliques, one for each vertex of $e$ and either $i'j$ or $ij'$ in the superscript. The size of each clique is at most $(n-1)(t-1)$, since the events can be identified with the $(n-1)$ edges incident to a vertex and the remaining $t-1$ trees.

Considering the symmetry of the dependency graph, we set the variables for all events equal to
$y^{i}_{ef} = y^{ij}_e = y$. The cluster expansion criteria will be satisfied if we set the parameters so that
$$ p \leq  \frac{y}{(1+(n-1)(t-1)y)^4 (1+(n-1)(q-1)y)^4} \leq \frac{y}{\sum_{I \subseteq \Gamma^+(E), I \in \Ind} y^I},$$
where $E$ denotes either $E^{i}_{ef}$ or $E^{ij}_e$.
The second inequality holds due to the structure of the neighborhood of each event that we described above.
We set $y = \beta p = 4\beta / n^2$ and assume $t \leq \gamma n, q \leq \gamma n$. The reader can verify that with the settings $\beta = (\frac87)^8$ and $\gamma = \frac{1}{32} (\frac78)^7$, we get $\frac{\beta}{(1+4\gamma \beta)^8} = 1$. Therefore,
$$ p \leq \frac{\beta p}{(1+4\gamma \beta)^8} \leq \frac{y}{(1+(n-1)(t-1)y)^4 (1+(n-1)(q-1)y)^4} $$
which verifies the assumption of \Theorem{cluster-with-slack}.
\Theorem{cluster-with-slack} implies that MaximalSetResample terminates after $O((\sum y^{i}_{ef} + \sum y^{ij}_e)^2)$ \ro calls
with high probability. The total number of events here is $O(t q n^2) = O(n^4)$ and for each event the respective variable is $y = O(1/n^2)$. Therefore, the expected number of \ro calls is $O(n^4)$. 
\end{proof}

\subsection{Rainbow matchings}

Given an edge-coloring of $K_{2n}$, a perfect matching is called rainbow if each of its edges has a distinct color. This can be viewed as a non-bipartite version of the problem of Latin transversals. It is known that given any {\em proper} $(2n-1)$-edge-coloring of $K_{2n}$ (where each color forms a perfect matching), there exists a rainbow perfect matching \cite{Woolbright}. However, finding rainbow matchings algorithmically is more difficult. Achlioptas and Iliopoulos \cite{Achlioptas} showed how to find a rainbow matching in $K_{2n}$ efficiently when each color appears on at most $\gamma n$ edges, $\gamma < \frac{1}{2e} \simeq 0.184$. Our result is that we can do this for $\gamma = \frac{27}{128} \simeq 0.211$. The improvement comes from the application of the ``cluster expansion" form of the local lemma, which is still efficient in our framework. (We note that an updated version of the Achlioptas-Iliopoulos framework \cite{AchlioptasI15} also contains this result.)

\begin{theorem}
\label{thm:rainbow-matching}
Given an edge-coloring of $K_{2n}$ where each color appears on at most $\frac{27}{128} n$ edges, a rainbow
perfect matching exists and can be found in $O(n^2)$ \ro calls with high probability.
\end{theorem}

In fact, we can find many disjoint rainbow matchings --- up to a linear number, if we replace $\frac{27}{128}$ above by a smaller constant. 

\begin{theorem}
\label{thm:rainbow-matchings}
Given an edge-coloring of $K_{2n}$ where each color appears on at most $\frac{7^7}{8^8} n$ edges, at least $\frac{7^7}{8^8} n$ edge-disjoint rainbow perfect matchings exist and can be found in $O(n^4)$ \ro calls whp.
\end{theorem}

We postpone the proof to Section~\ref{sec:Latin}, since it follows from our result for Latin transversals.

\begin{proof}[Proof of Theorem~\ref{thm:rainbow-matching}]
We apply our algorithm in the setting of uniformly random perfect matchings $M \subset K_{2n}$, with
the following bad events (identical to the setup in \cite{Achlioptas}): For every pair of edges
$e,f$ of the same color, $E_{ef}$ occurs if $\{e,f\} \subset M$. If no bad event $E_{ef}$ occurs
then $M$ is a rainbow matching. We also define the following dependency graph: $E_{ef} \sim
E_{e'f'}$ unless $e,f,e',f'$ are four disjoint edges. Note that this is more conservative than the
dependency graph we considered in Section~\ref{sec:resample-matchings}, where two events are only
connected if they do not form a matching together. The more conservative definition will simplify
our analysis. In any case, our \ro is consistent with this \NDG in the sense
that resampling $E_{ef}$ can only cause new events $E_{e'f'}$ such that $E_{ef} \sim E_{e'f'}$. We
show that this setup satisfies the criteria of the cluster expansion lemma.

Let $q = \frac{27}{128} n$,  $p = \frac{1}{(2n-1)(2n-3)}$ and $y = (\frac43)^4 p$.
Consider the neighborhood of a bad event $\Gamma(E_{ef})$. It contains all events $E_{e'f'}$ such that there is some intersection among the edges $e,f,e',f'$. Such events can be partitioned into $4$ cliques: for each vertex $v \in e \cup f$, let ${\cal Q}_v$ denote all the events $E_{e'f'}$ such that $v \in e'$ and $f'$ has the same color as $e'$. The number of edges $e'$ incident to $v$ is $2n-1$, and for each of them, the number of other edges of the same color is by assumption at most $q-1$. Therefore, the size of ${\cal Q}_v$ is at most $(q-1)(2n-1)$.

In the following, we use the short-hand notation $y^I = \prod_{i \in I} y_i$. 
Consider the assumptions of the cluster expansion lemma: for each event $E_{ef}$, we should have
 $$ \Pr[E_{ef}] \leq \frac{y_{ef}}{\sum_{I \subseteq \Gamma^+(E_{ef}), I \in \Ind} y^I}.$$
We have $\Pr[E_{ef}] = p = \frac{1}{(2n-1)(2n-3)}$. By symmetry, we set all the variables $y_{ef}$ to the same value, $y_{ef} = y = (\frac43)^4 p$.
Note that an independent subset of $\Gamma^+(E_{ef})$ can contain at most 1 event from each clique ${\cal Q}_v$. (The event $E_{ef}$ itself is also contained in these cliques.) Therefore,
$$ \sum_{I \subseteq \Gamma^+(E_{ef}), I \in \Ind} y^I \leq \prod_{v \in e \cup f} (1 + \sum_{E_{e'f'} \in {\cal Q}_v} y_{e'f'}) \leq \left(1 + (q-1)(2n-1) y \right)^4.$$
The reader can verify that
$\sum_{I \subseteq \Gamma^+(E_{ef}), I \in \Ind} y^I 
\leq (1 + (q-1)(2n-1) y)^4 \leq (1 + \frac{27}{64} n^2 (\frac43)^4 / (2n)^2)^4 = (\frac43)^4$.
Therefore,
$$ \frac{y}{\sum_{I \subseteq \Gamma^+(E_{ef}), I \in \Ind} y^I} \geq p$$
which is the assumption of \Theorem{cluster-with-slack}.
By \Theorem{cluster-with-slack}, MaximalSetResample with the \ro for matchings and the dependency graph defined above will find a rainbow perfect matching in time $O(\sum_{E_{ef}} y_{ef} \sum_{E_{ef}} \log(1 + y_{ef})) = O((\sum_{E_{ef}} y_{ef})^2) $ with high probability. The number of bad events $E_{ef}$ is $O(n^3)$, because each color class has $O(n)$ edges so the number of edge pairs of equal color is $O(n^3)$. We have $y_{ef} = O(1/n^2)$, and hence the total number of resamplings is $O(n^2)$ with high probability.
\end{proof}

\subsection{Latin transversals}
\label{sec:Latin}

A Latin transversal in an $n \times n$ matrix $A$ is a permutation $\pi \in S_n$ such that the
entries $A_{i,\pi(i)}$ (``colors") are distinct for $i = 1,2,\ldots,n$. In other words, it is a set
of distinct entries, exactly one in each row and one in each column. It is easy to see that this is
equivalent to a bipartite version of the rainbow matching problem: $A_{ij}$ is the color of the edge
$(i,j)$ and we are looking for a perfect bipartite matching where no color appears twice. It is a
classical application of the Lov\'asz Local Lemma that if no color appears more than $\frac{1}{4e}
n$ times in $A$ then there exists a Latin transversal \cite{ErdosSpencer}. An improvement of this
result is that if no color appears more than $\frac{27}{256} n$ times in $A$ then a Latin
transversal exists \cite{Bissacot}; this paper introduced the ``cluster expansion" strengthening of
the local lemma. (Note that $\frac{27}{256} = \frac{3^3}{4^4}$.) These results were made algorithmically efficient by the work of Harris and Srinivasan \cite{HarrisS14}.

Beyond finding one Latin transversal, one can ask whether there exist multiple disjoint Latin transversals. A remarkable existential result was proved by Alon, Spencer and Tetali \cite{AlonSpencerTetali}: If $n = 2^k$ and each color appears in $A$ at most $\epsilon n$ times ($\epsilon = 10^{-10^{10}}$ in their proof), then $A$ can be partitioned into $n$ disjoint Latin transversals. Here, we show how to find a linear number of Latin transversals algorithmically.

\begin{theorem}
\label{thm:Latin-transversals}
For any $n \times n$ matrix $A$ where each color appears at most $\frac{7^7}{8^8} n$ times, there exist at least $\frac{7^7}{8^8} n$ disjoint Latin transversals, and they can be found in $O(n^4)$ \ro calls w.h.p.
\end{theorem}

We note that again, if there is constant multiplicative slack in the assumption on color appearances, the number of resamplings improves to $O(n^2)$. 
This also implies Theorem~\ref{thm:rainbow-matchings} as a special case: For an edge-coloring of $K_{2n}$ where no color appears more than $\frac{7^7}{8^8} n$ times, let us label the vertices arbitrarily $(u_1,\ldots,u_n,v_1,\ldots,v_n)$ construct a matrix $A$ where $A_{ij}$ is the color of the edge $(u_i,v_j)$. If no color appears more than $\frac{7^7}{8^8} n$ times, by Theorem~\ref{thm:Latin-transversals} we can find $\frac{7^7}{8^8} n$ Latin transversals; these correspond to rainbow matchings in $K_{2n}$.

Our approach to proving Theorem~\ref{thm:Latin-transversals} is similar to the proof of Theorem~\ref{thm:rainbow-trees}: sample $\frac{7^7}{8^8} n$ independently random permutations and hope that they will be (a) disjoint, and (b) Latin. For reasons similar to Theorem~\ref{thm:rainbow-trees}, the local lemma works out and our framework makes this algorithmic.

\begin{proof}
Let $t = \frac{7^7}{8^8} n$ and let $\pi_1,\ldots,\pi_t$ be independently random permutations on $[n]$. 
We consider the following two types of bad events:
\begin{itemize}
\item $E^{i}_{ef}$: For each $i \in [t]$ and $e=(u,v), f = (x,y) \in [n] \times [n]$ such that $u \neq v, x \neq y, A_{uv} = A_{xy}$, the event $E^{i}_{ef}$ occurs if $\pi_i(u) = v$ and $\pi_i(x) = y$;
\item $E^{ij}_e$: For each $i \neq j \in [t]$ and $e=(u,v) \in [n] \times [n]$, the event $E^{ij}_e$ occurs if $\pi_i(u) = \pi_j(u) = v$.
\end{itemize}
Clearly, if none of these events occurs then the permutations $\pi_1,\ldots,\pi_t$ correspond to pairwise disjoint Latin transversals. 
The probability of a bad event of the first type is $\Pr[E^{i}_{ef}] = \frac{1}{n(n-1)}$ and the probability for the second type is $\Pr[E^{ij}_e] = \frac{1}{n^2}$. Thus the probability of each bad event is at most $p = \frac{1}{n(n-1)}$.

It will be convenient to think of the pairs $e = (x,y) \in [n] \times [n]$ as edges in a bipartite
complete graph.  As we proved in Section~\ref{sec:resample-permutations}, the \ro for permutations
is consistent with the following \NDG graph. 
\begin{itemize}
\item $E^{i}_{ef} \sim E^{i}_{e'f'}$ whenever there is some intersection between the edges $e,f$ and $e',f'$;
\item $E^{i}_{ef},E^{j}_{ef} \sim E^{ij}_{e'}$ whenever there is some intersection between  $e'$ and $e,f$;
\item $E^{ij}_{e} \sim E^{ij'}_{e'}, E^{i'j}_{e'}$ whenever $e'$ intersects $e$.
\end{itemize}
By Lemma~\ref{lem:dependency}, the \ro for a given event never causes a new event except in its neighborhood.

Let us now verify the cluster expansion criteria. The counting here is quite similar to the proof of Theorem~\ref{thm:rainbow-trees}, so we skim over some details. The neighborhood of each event $E^{i}_{ef}$ consist of $8$ cliques: $4$ cliques of events of type $E^{i}_{e'f'}$ and $4$ cliques of events of type $E^{ij}_e$, corresponding in each case to the 4 vertices of $e \cup f$. In the first case, each clique has at most $n(q-1)$ events, determined by selecting an incident edge and another edge of the same color. In the second case, each clique has at most $n(t-1)$ events, determined by selecting an incident edge and another permutation.

The neighborhood of each event $E^{ij}_{e}$ also consists of $8$ cliques: $4$ cliques of events $E^{i}_{e'f'}$ or $E^{j}_{e'f'}$, corresponding to the choice of either $i$ or $j$ in the superscript, and one of the two vertices of $e$. The size of each clique is at most $n(q-1)$, determined by choosing an incident edge and another edge of the same color. Then, we have $4$ cliques of events $E^{ij'}_{e'}$ or $E^{i'j}_{e'}$, determined by switching either $i'$ or $j'$ in the superscript, and choosing one of the vertices of $e$. The size of each clique is at most $n(t-1)$, determined by choosing an incident edge and a new permutation in the superscript.

As a consequence, the cluster expansion criterion here is almost exactly the same as in the case of Theorem~\ref{thm:rainbow-trees}:
$$ p \leq \frac{y}{(1 + n(t-1)y)^4 (1 + n(q-1)y)^4}.$$
We have $p = \frac{1}{n(n-1)}$ here and we set $y = \beta p$. For $t,q \leq \gamma n$, it's enough
to satisfy $\frac{\beta}{(1+\beta \gamma)^8} \geq 1$, which is achieved by $\beta = (\frac87)^8$ and
$\gamma = \frac{7^7}{8^8}$. Therefore, \Theorem{cluster-with-slack} implies that
MaximalSetResample will terminate within $O((\sum y^i_{ef} + \sum y^{ij}_e)^2) = O(n^4)$ \ro calls with high probability.
\end{proof}

\section{Analysis of the algorithm}
\SectionName{analysis}

Here we provide the analysis of our algorithm and the proofs of our main theorems.
In \Section{stable-set-sequences}, we begin with the basic notions necessary for our analysis and a coupling argument which forms the basis of all our algorithmic results.
In \Section{LLLslack}, we prove a weaker form of \Theorem{LLL-tight-result}
under the assumption that the \eqref{eq:GLL} criterion holds with some slack.
In \Section{shearer}, we introduce the independence polynomial of a graph and summarize its fundamental properties that are important for our analysis. 
In \Section{shearer-slack}, we prove that our algorithm is efficient if Shearer's criterion is
satisfied with an $\epsilon$ slack. In \Section{shearer-automatic-slack}, we show that in some sense
this assumption is not necessary, because every point satisfying Shearer's criterion has some slack
available, and we quantify how large this slack is. Finally, we return to the weaker (but more
practical) variants of the local lemma: the \eqref{eq:GLL} and \eqref{eq:CLL} criteria. We present
new combinatorial connections between these criteria and Shearer's criterion, which in turn imply
our main results on the efficiency of our algorithm under the \eqref{eq:GLL}  and \eqref{eq:CLL}
criteria (in Sections \ref{sec:LLLimpliesShearer} and \ref{sec:cluster}, respectively).

\subsection{Stable set sequences and the coupling argument}
\SectionName{stable-set-sequences}

An important notion in our analysis is that of {\em stable set sequences}.
We note that this concept originated in the work of Kolipaka and Szegedy \cite{Kolipaka}
which builds on Shearer's work \cite{Shearer}.
There are some similarities but also differences in how this concept is applied here:
most notably, our stable set sequences grow forward in time, while the stable set sequences in \cite{Kolipaka} grow backward in time (which is similar to the Moser-Tardos analysis \cite{MoserTardos}).

\begin{definition}
One execution of the outer repeat loop in \MSR is called an \emph{iteration}.
For a sequence of non-empty sets $\cI = (I_1,\ldots,I_t)$, we say that the algorithm \emph{follows}
$\cI$ if $I_s$ is the set resampled in iteration $s$ for $1 \leq s < t$, and $I_t$ is a set of the
first $m$ events resampled in iteration $t$ for some $m \geq 1$ (a prefix of the maximal independent set constructed in iteration $t$).
\end{definition}

Recall that $\Ind=\Ind(G)$ denotes the independent sets (including the empty set)
in the graph under consideration.

\begin{definition}
$\cI = (I_1,I_2,\ldots,I_t)$ is called a \emph{stable set sequence} if $I_1,\ldots,I_t \in \Ind(G)$ and
$I_{s+1} \subseteq \Gamma^+(I_s)$ for each $1 \leq s < t$. We call the sequence $\cI$ \emph{proper} if each independent set $I_s$ is nonempty.
\end{definition}

Note that if $I_s = \emptyset$ for some $s$, then $I_t = \emptyset$ for all $t > s$. Therefore, the
nonempty sets always form a prefix of the stable set sequence. Formally, we consider an empty sequence also a stable set sequence, of length $0$.

\begin{lemma}
If \MSR follows a sequence $\cJ = (J_1,\ldots,J_t)$, then $\cJ$ is a stable set sequence.
\end{lemma}

\begin{proof}
By construction, the set $J_s$ chosen in each iteration is independent in $G$. For each $i \in J_s$,
we execute the \ro $r_i$. Recall that $r_i$ executed on a satisfied event $E_i$ can only cause new events in the neighborhood $\Gamma^+(i)$ (and this neighborhood is not explored again until the following iteration).
Since $J_s$ is a maximal independent set of satisfied events, all the events satisfied in the following iteration are neighbors of some event in $J_s$, i.e., $J_{s+1} \subseteq \Gamma^+(J_s)$. In the last iteration, this also holds for a subset of the resampled events.
\end{proof}

We use the following notation:
For $i \in [n]$, $p_i = \Pr_\mu[E_i]$. For $S \subseteq [n]$, $p^S = \prod_{i \in S} p_i$. 
For a stable set sequence $\cI=(I_1,\ldots,I_t)$, $p_{\cI} = \prod_{s=1}^{t} p^{I_s}$.
We relate stable set sequences to executions of the algorithm by the following coupling argument.
Although the use of stable set sequences is inspired by \cite{Kolipaka}, their coupling argument is
different due to its backward-looking nature (similar to \cite{MoserTardos}),
and their restriction to the variable model.

\begin{lemma}
\label{lem:prod-bound}
For any proper stable set sequence $\cI = (I_1, I_2, \ldots, I_t)$, the probability that the
\MSR algorithm follows $\cI$ is at most $p_{\cI}$.
\end{lemma}

\begin{proof}
Given $\cI=(I_1,I_2,\ldots,I_t)$, let us consider the following ``$\cI$-checking" random
process. We start with a random state $\omega \sim \mu$. In iteration $s$, we process the events of
$I_s$ in the ascending order of their indices. For each $i \in I_s$, we check whether $\omega$
satisfies $E_i$; if not, we terminate. Otherwise, we apply the \ro $r_i$ and replace $\omega$ by
$r_i(\omega)$. We continue for $s=1,2,\ldots,t$. We say that the $\cI$-checking process succeeds if every event is satisfied when checked and the process runs until the end.

By induction, the state $\omega$ after each \ro call is distributed according to $\mu$:  Assuming this
was true in the previous step and conditioned on $E_i$ satisfied, we have $\omega \sim \mu|_{E_i}$.
By assumption, the \ro $r_i$ removes this conditioning and produces again a random state $r_i(\omega) \sim \mu$. 
Therefore, whenever we check event $E_i$, it is satisfied with probability $\Pr_{\mu}[E_i]$
(conditioned on the past). By a telescoping product of conditional probabilities, the probability
that the $\cI$-checking process succeeds is exactly
$\prod_{s=1}^{t} \prod_{i \in I_s} \Pr_\mu[E_i] = \prod_{s=1}^{t} p^{I_s} = p_\cI$.

To conclude, we argue that the probability that \MSR follows the sequence
$\cI$ is at most the probability that the $\cI$-checking process
succeeds. To see this, suppose that we couple \MSR and the $\cI$-checking
process, so they use the same source of randomness. In each iteration, if \MSR
includes $i$ in $J_t$, it means that $E_i$ is satisfied. Both procedures apply the \ro $r_I(\omega)$
and by coupling the distribution in the next iteration is the same. Therefore, the event that
\MSR follows the sequence $\cI$ is contained in the event that the $\cI$-checking
process succeeds, which happens with probability $p_\cI$.
\end{proof}

We emphasize that we do \emph{not} claim that the distribution of the current state $\omega \in
\Omega$ is $\mu$
after each \ro call performed by the \MSR algorithm. This would mean that the algorithm is
not making any progress in its search for a state avoiding all events.
It is only the $\cI$-checking process that has this property.

\begin{definition}
Let $\Stab$ denote the set of all stable set sequences and $\Prop$ the set of proper stable set sequences. Let us denote by $\Stab_\ell$ the set of stable set sequences $(I_1,\ldots,I_\ell)$ of length $\ell$, and by $\Stab_\ell(J)$ the subset of $\Stab_\ell$ such that the first set in the sequence is $J$. 
Similarly, denote by $\Prop_\ell$ the set of proper stable set sequences of length $\ell$, and by $\Prop(J)$ the subset of $\Prop$ such that the first set in the sequence is $J$.
For $\cI = (I_1,\ldots,I_t) \in \Prop$, let us call $\sigma(\cI) = \sum_{s=1}^{t} |I_s|$ the total size of the sequence.
\end{definition}

\begin{lemma}
\label{lem:iteration-bound}
The probability that \MSR runs for at least $\ell$ iterations is at most $\sum_{\cI \in \Prop_\ell} p_\cI$. The probability that \MSR resamples at least $s$ events is at most $\sum_{\cI \in \Prop: \sigma(\cI)=s} p_\cI$.
\end{lemma}

\begin{proof}
If the algorithm runs for at least $\ell$ iterations, it means that it follows some proper sequence $\cI = (I_1,I_2,\ldots,I_\ell)$. 
By Lemma~\ref{lem:prod-bound}, the probability that the algorithm follows a particular stable set sequence $\cI$ is at most $p_\cI$. By the union bound, the probability that the algorithm runs for at least $\ell$ iterations is at most $\sum_{\cI=(I_1,\ldots,I_\ell) \in \Prop} p_\cI$.

Similarly, if the algorithm resamples at least $s$ events, it means that it follows some proper sequence $\cI$ of total size $\sigma(\cI) = s$. By the union bound, the probability of resampling at least $s$ events is upper-bounded by $\sum_{\cI \in \Prop: \sigma(\cI)=s} p_\cI$.
\end{proof}

We note that these bounds could be larger than $1$ and thus vacuous. The events that ``the algorithm follows $\cI = (I_1,\ldots,I_\ell)$" are disjoint for different sequences of fixed total size $\sigma(\cI)$, while they could overlap for a fixed length $\ell$ (because we can take $I_\ell$ to be different prefixes of the sequence of events resampled in iteration $t$). In any case, the upper bound of $p_\cI$ on each of the events could be quite loose.

\subsection{A simple analysis: the General Lov\'asz Lemma criterion, with slack}
\SectionName{LLLslack}

In this section we will analyze the algorithm under the assumption that the
\eqref{eq:GLL} criterion holds with some ``slack''.
This idea of exploiting slack has appeared in previous work, e.g., \cite{MoserTardos,CGH,HSS,Kolipaka}.
This analysis proves only a weaker form of \Theorem{LLL-tight-result}.
The full proof, which removes the assumption of slack, appears in \Section{LLLimpliesShearer}.

To begin, let us prove the following (crude) bound on the expected number of iterations.
We note that this bound is typically exponentially large.

\begin{lemma}
\label{lem:crude-bound}
Provided that the $p_i$ satisfy the \eqref{eq:GLL} criterion, $p_i \leq x_i \prod_{j \in \Gamma(i)} (1-x_j)$, we have
$$ \sum_{\cI \in \Prop} p_\cI ~\leq~ \prod_{i=1}^{n} \frac{1}{1-x_i}.$$
\end{lemma}

\begin{proof}
It will be convenient to work with sequences of fixed length, where we pad by empty sets if necessary. Note that by definition this does not change the value of $p_\cI$: e.g., $p_{(I_1,I_2)} = p_{(I_1,I_2,\emptyset,\ldots,\emptyset)}$. Recall that $\Stab_\ell(J)$ denotes the set of all stable set sequences of length $\ell$ where the first set is $J$. We show the following statement by induction on $\ell$: For any $J \in \Ind$ and any $\ell \geq 1$,
\begin{equation}
\label{eq:stab-induction}
\sum_{\cI  \in \Stab_\ell(J)} p_\cI ~\leq~ \prod_{j \in J} \frac{x_j}{1-x_j}.
\end{equation}

This is true for $\ell=1$, since $p_{(J)} = p^J \leq \prod_{j \in J} x_j$ by the LLL assumption.
Let us consider the expression for $\ell+1$. We have
$$ \sum_{\cI' \in \Stab_{\ell+1}(J)} p_{\cI'}
 ~=~ p^{J} \sum_{J' \subseteq \Gamma^+(J)} \sum_{\cI \in \Stab_\ell(J')} p_\cI
 ~\leq~ p^{J} \sum_{J' \subseteq \Gamma^+(J)} \prod_{i \in J'} \frac{x_i}{1-x_i} $$
by the inductive hypothesis. This can be simplified using the following identity:
\begin{equation}
\EquationName{trivialidentity}
\prod_{i \in \Gamma^+(J)} (1+\alpha_i)
 = \sum_{I_1 \subseteq \Gamma^+(J)} \prod_{i \in I_1} \alpha_i.
\end{equation}
We use this with $\alpha_i = \frac{x_i}{1-x_i}$. Therefore,
$$ \sum_{\cI' \in \Stab_{\ell+1}(J)} p_{\cI'} 
 ~\leq~ p^{J} \prod_{i \in \Gamma^+(J)} \left(1 + \frac{x_i}{1-x_i} \right)
  ~=~ p^{J} \prod_{i \in \Gamma^+(J)} \frac{1}{1-x_i}.$$
Now we use the LLL assumption: 
$$p^{J} ~=~ \prod_{i \in J} p_i \leq \prod_{i \in J} \left( x_i \prod_{j \in \Gamma(i)} (1-x_j) \right)
 ~\leq~ \prod_{i \in J} x_i \prod_{j \in \Gamma^+(J) \setminus J} (1-x_j) $$
because each element of $\Gamma^+(J) \setminus J$ appears in $\Gamma(i)$ for at least one $i \in J$.
We conclude that
$$ \sum_{\cI' \in \Stab_\ell(J)} p_{\cI'} 
 ~\leq~ \prod_{i \in J} x_i \prod_{j \in \Gamma^+(J) \setminus J} (1-x_j) \cdot \prod_{i' \in \Gamma^+(J)} \frac{1}{1-x_{i'}}
 ~=~ \prod_{i \in J} \frac{x_j}{1-x_j}. $$
This proves (\ref{eq:stab-induction}).

Adding up over all sets $J \subseteq [n]$, we again use \Equation{trivialidentity} to obtain
$$ \sum_{\cI \in \Stab_\ell} p_\cI
 ~\leq~ \sum_{J \subseteq [n]} \prod_{j \in J} \frac{x_j}{1-x_j}
 ~=~ \prod_{i=1}^{n} \left( 1 + \frac{x_i}{1-x_i} \right)
 ~=~ \prod_{i=1}^{n} \frac{1}{1-x_i}.$$
As we argued above, this can be written equivalently as
$$ \sum_{k=1}^{\ell} \sum_{\cI \in \Prop} p_\cI ~\leq~ \prod_{i=1}^{n} \frac{1}{1-x_i}.$$
Since this is true for every $\ell$, and the left-hand-side is non-increasing in $\ell$,
the sequence as $\ell \rightarrow \infty$ has a limit and the bound still holds in the limit.
\end{proof}

The following is our first concrete result: our algorithm is efficient if \eqref{eq:GLL} is satisfied with a slack.

\begin{theorem}
\TheoremName{GLL-with-slack}
If \eqref{eq:GLL} is satisfied with a slack of $\epsilon$, i.e.
$$ \Pr_\mu[E_i] ~\leq~ (1-\epsilon) x_i \prod_{j \in \Gamma(i)} (1-x_j) $$
then with probability $1-e^{-t}$ \MSR resamples at most $\frac{1}{\epsilon} (t + \sum_{i=1}^{n} \ln \frac{1}{1-x_i})$ events.
\end{theorem}

\begin{proof}
By Lemma~\ref{lem:iteration-bound}, the probability that \MSR resamples more than $s$ events is at most $\sum_{\cI \in \Prop:\sigma(\cI)=\lceil s \rceil} p_\cI$ where $p_\cI$ is the product of $p_i = \Pr_\mu[E_i]$ over all events in the sequence $\cI$. By the slack assumption, we have $p_i \leq (1-\epsilon) p'_i$ and $p_\cI \leq (1-\epsilon)^{\sigma(\cI)} p'_\cI$, where $p'_i = x_i \prod_{j \in \Gamma(i)} (1-x_j)$. Using Lemma~\ref{lem:crude-bound}, we obtain
$$ \sumstack{\cI \in \Prop \\ \sigma(\cI)=\lceil s \rceil} p_\cI
 ~\leq~ (1-\epsilon)^s \sum_{\cI \in \Prop} p'_\cI
 ~\leq~ e^{-\epsilon s} \prod_{i=1}^{n} \frac{1}{1-x_i}. $$
For $s = \frac{1}{\epsilon} (t + \sum_{i=1}^{n} \ln \frac{1}{1-x_i})$, we obtain
$$ \sumstack{\cI \in \Prop \\ \sigma(\cI)=\lceil s \rceil} p_\cI
 ~\leq~ e^{-\epsilon s} \prod_{i=1}^{n} \frac{1}{1-x_i}
 ~\leq~ e^{-t}.
$$
Therefore, the probability of resampling more than $s$ events is at most $e^{-t}$.
\end{proof}

\subsection{Preliminaries on Shearer's criterion}
\SectionName{shearer}

In this section we discuss a strong version of the local lemma due to Shearer \cite{Shearer}.
Shearer's lemma is based on certain forms of the multivariate independence polynomial.
We recall that $p^I$ denotes $\prod_{i \in I} p_i$.

\begin{definition}
Given a graph $G$ and values $p_1,\ldots,p_n$, define for each $S \subseteq [n]$
\begin{equation}
\EquationName{qdef}
q_S ~=~ q_S(p) ~=~ \sumstack{I \in \Ind \\ S \subseteq I} (-1)^{|I \setminus S|} p^I .
\end{equation}
\end{definition}

Note that $q_S=0$ for $S \notin \Ind$.
An alternative form of these polynomials that is also useful is obtained by summing over subsets of $S$.

\newcommand{\qdown}{\breve{q}}

\begin{definition}
Given a graph $G$ and values $p_1,\ldots,p_n$, define
$$ \qdown_S ~=~ \qdown_S(p) ~=~ \sumstack{I \in \Ind \\ I \subseteq S} (-1)^{|I|} p^I.$$
\end{definition}

The following set plays a fundamental role.

\begin{definition}
\DefinitionName{Shearer-region}
Given a graph $G$, the \newterm{Shearer region} is the semialgebraic set
\begin{subequations}
\begin{align}
\EquationName{SR1}
\cS &~=~ \setst{ p \in (0,1)^n }{ \forall I \in \Ind ,\: q_I(p) > 0 }
    \\\EquationName{SR2}
    &~=~ \setst{ p \in (0,1)^n }{ \forall S \subseteq [n] ,\: \qdown_S(p) > 0 }
\end{align}
\end{subequations}
\end{definition}

The equivalence between \eqref{eq:SR1} and \eqref{eq:SR2} is proven below in \Claim{SRequiv}.

Shearer's Lemma can be stated as follows.

\begin{lemma}[Shearer \cite{Shearer}]
\LemmaName{Shearer-Lemma}
Let $G$ be a lopsidependency graph
for the events $E_1,\ldots,E_n$.
Let $p_i = \Pr_\mu[E_i] \in (0,1)$. 
If $p \in \cS$ then $\Pr_\mu[\bigcap_{i=1}^{n} \overline{E_i}] \geq q_\emptyset$.
\end{lemma}

It is known that Shearer's Lemma implies \Theorem{LLL}, as we will see in \Section{LLLimpliesShearer}, and in fact gives the tight criterion under which all events can be avoided for a given dependency graph $G$.
The polynomials $q_S(p)$ and $\qdown_S(p)$ have a natural interpretation in the Shearer region:
There is a ``tight instance" where $q_S(p)$ is the probability that the set of occurring events is exactly $S$,
and $\qdown_S(p)$ is the probability that none of the events in $S$ occur. In particular, $q_\emptyset(p) = \qdown_{[n]}(p)$
is exactly the probability that no event occurs. (See \cite{Shearer} for more details.)

\subsubsection{Properties of independence polynomials}

In this section we summarize some of the important properties of these polynomials, most of which may be found in earlier work.
Since some of the proofs are not easy to recover due to different notation and/or their analytic nature (in case of \cite{ScottSokal}), we provide short combinatorial proofs for completeness.

\begin{claim}[The ``fundamental identity''. {\protect Shearer \cite{Shearer}, Scott-Sokal \cite[Eq.~(3.5)]{ScottSokal}}]
\ClaimName{breve1}
For any $a \in S$, we have
$$\qdown_S ~=~ \qdown_{S \setminus \set{a}} \,-\, p_a \cdot \qdown_{S \setminus \Gamma^+(a)}.$$
\end{claim}

\begin{proof}
Every independent set $I \subseteq S$ either contains $a$ or does not. In addition, if $a \in I$ then $I$ is independent iff $I \setminus \{a\}$ is an independent subset of $S \setminus \Gamma^+(a)$.
\end{proof}

\begin{claim}[{\protect Shearer \cite{Shearer}, Scott-Sokal \cite[Eq.~(2.52)]{ScottSokal}}]
\ClaimName{breve2}
For every $S \subseteq [n]$,
$$ \qdown_S ~=~ \sum_{Y \subseteq [n] \setminus S} q_Y.$$
\end{claim}

\begin{proof}
By definition of $q_Y$,
$$
\sum_{Y \subseteq [n] \setminus S} q_Y
 ~=~ \sum_{Y \subseteq [n] \setminus S} \sumstack{I \in \Ind \\ Y \subseteq I} (-1)^{|I \setminus Y|} p^I
 ~=~ \sum_{I \in \Ind} p^I \sum_{Y \subseteq I \setminus S} (-1)^{|I \setminus Y|}.
$$
If $I \setminus S \neq \emptyset$ then the last alternating sum is zero. Therefore, the sum
simplifies to $\sum_{I \in \Ind: I \subseteq S}  (-1)^{|I|} p^I = \qdown_S$ as required.
\end{proof}

\begin{claim}[{\protect Shearer \cite{Shearer}}]
\ClaimName{q-sum}
$$ \sum_{J \in \Ind} q_J ~=~ \sum_{S \subseteq [n]} q_S ~=~ 1.$$
\end{claim}

\begin{proof}
Set $S = \emptyset$ in \Claim{breve2} and use the fact that $\qdown_\emptyset = 1$.
\end{proof}

\begin{claim}[{\protect Scott-Sokal \cite[Eq.~(2.48)]{ScottSokal}}]
\ClaimName{breve3}
For $I \in \Ind$, $$q_I ~=~ p^I \cdot \qdown_{[n] \setminus \Gamma^+(I)}.$$
\end{claim}

\begin{proof}
Given $I \in \Ind$, each independent set $J \supseteq I$ can be written uniquely as $J = I \union K$
where $K$ is independent and $K \intersect \Gamma^+(I) = \emptyset$. So,
$$
q_I
    ~=~ \sum_{J \in \Ind: I \subseteq J} (-1)^{\card{J \setminus I}} p^J
    ~=~ p^I \sumstack{K \in \Ind \\ K \subseteq [n] \setminus \Gamma^+(I)} (-1)^{\card{K}} p^K
    ~=~ p^I \cdot \qdown_{[n] \setminus \Gamma^+(I)}.
$$
\end{proof}

\begin{lemma}[{\protect Kolipaka-Szegedy \cite[Lemma~15]{Kolipaka}}]
\label{lem:q-expansion}
For any $I \in \Ind$
$$ q_I ~=~ p^I \cdot \sum_{S \subseteq \Gamma^+(I)} q_S.$$
\end{lemma}

\begin{proof}
By \Claim{breve3} and \Claim{breve2}, we have
$
q_I ~=~ p^I \cdot \qdown_{[n] \setminus \Gamma^+(I)}
    ~=~ p^I \sum_{S \subseteq \Gamma^+(I)} q_S,
$
as required.
\end{proof}

\begin{claim}[Simultaneous positivity of $q_S$ and $\qdown_S$]
\ClaimName{QAndQDown}
Assume that $p \in [0,1]^n$. Then
\begin{align}
\EquationName{UpImpliesDown}
q_I \geq 0 ~~\forall I \in \Ind
&\qquad\implies\qquad
\qdown_S \geq q_\emptyset ~~\forall S \subseteq [n]
\\
\EquationName{DownImpliesUp}
\qdown_S \geq 0 ~~\forall S \subseteq [n]
&\qquad\implies\qquad
q_I \geq p^{[n]} \cdot \qdown_{[n]} ~~\forall I \in \Ind.
\end{align}
\end{claim}

\begin{proof}
\eqref{eq:UpImpliesDown} follows from \Claim{breve2} (since $q_Y=0$ for $Y \notin \Ind$).
To see \eqref{eq:DownImpliesUp},
first note that $q_I \geq 0$ for all $I \in \Ind$, by \Claim{breve3}.
Consequently, by \Claim{breve2}, $\qdown_{[n]} = \min_S \, \qdown_S$.
Clearly, $p^{[n]} = \min_I \, p^I$.
It follows from \Claim{breve3} again that
$q_I = p^I \cdot \qdown_{[n] \setminus \Gamma^+(I)} \geq p^{[n]} \cdot \qdown_{[n]}$.
\end{proof}

\begin{claim}
\ClaimName{SRequiv}
The two characterizations of the Shearer region,
\eqref{eq:SR1} and \eqref{eq:SR2}, are equivalent.
\end{claim}
\begin{proof}
By \Claim{QAndQDown}, if $q_\emptyset > 0$ and $q_S \geq 0 \ \forall S \subseteq [n]$, then $\qdown_S > 0$ for all $S \subseteq [n]$. Conversely, if $\qdown_S > 0$ for all $S \subseteq [n]$, then $q_I \geq p^{[n]} \qdown_{[n]} > 0$ for all $I \in \Ind$. 
\end{proof}

\begin{claim}[Monotonicity of $\qdown$, {\protect Scott-Sokal \cite[Theorem~2.10]{ScottSokal}}]
\ClaimName{brevemonotone}
Let $p \in [0,1]^n$.
$$
\qdown_S(p) \geq 0 ~~\forall S \subseteq [n]
\qquad\implies\qquad
\qdown_S(p') \geq \qdown_S(p) \quad\forall 0 \leq p' \leq p ,\: \forall S \subseteq [n].
$$
\end{claim}
\begin{proof}
First consider the case that $p$ and $p'$ differ only in coordinate $i$.
For any $S \subseteq [n]$, \Claim{breve1} implies that
$\frac{\partial}{\partial p_i} \qdown_S(p) = - \qdown_{S \setminus \Gamma^+(i)}(p)$
and $\frac{\partial^2}{\partial p_i^2} \qdown_S = 0$.
Thus,
$$
\qdown_S(p')
 ~=~ \qdown_S(p) + (p_i-p'_i) \cdot \qdown_{S \setminus \Gamma^+(i)}(p)
 ~\geq~ \qdown_S(p).
$$
The case that $p'$ and $p$ differ in multiple coordinates is handled by induction.
\end{proof}

\begin{claim}[Log-submodularity of $\qdown_S$, {\protect Scott-Sokal \cite[Corollary~2.27]{ScottSokal}}]
\ClaimName{breve-submodular}
For any $p \in \cS$ and $A,B \subseteq [n]$, we have
$\qdown_A \cdot \qdown_B \geq \qdown_{A \union B} \cdot \qdown_{A \intersect B}$.
\end{claim}

\begin{proof}
We claim that for any $a \in S \subseteq T$, we have 
\begin{equation}
\EquationName{breve-ratios}
\frac{\qdown_S}{\qdown_{S \setminus \{a\}}} ~\geq~ \frac{\qdown_T}{\qdown_{T \setminus \{a\}}}.
\end{equation}
By induction, this implies that for any $R \subseteq S$, $\frac{\qdown_S}{\qdown_{S \setminus R}} \geq \frac{\qdown_T}{\qdown_{T \setminus R}}$.
We obtain the claim above by setting $S = A$, $T = A \cup B$, and $R = A \setminus B$.

We prove \Equation{breve-ratios} again by induction, on $|T|$. For $|T|=1$, the statement is trivial. Let $|T|>1$. By \Claim{breve1}, we have
$$ \qdown_S ~=~ \qdown_{S \setminus \{a\}} - p_a \qdown_{S \setminus \Gamma^+(a)} $$
and
$$ \qdown_T ~=~ \qdown_{T \setminus \{a\}} - p_a \qdown_{T \setminus \Gamma^+(a)}. $$
Let us denote $S \cap \Gamma^+(a) = \{ a, s_1,\ldots,s_k \}$. We apply \Equation{breve-ratios} to strict subsets of $S$ and $T$, to obtain
$$ \frac{\qdown_{S \setminus \Gamma^+(a)}}{\qdown_{S \setminus \{a\}}}
 ~=~ \prod_{i=1}^{k} \frac{\qdown_{S \setminus \{a,s_1,\ldots,s_{i-1},s_i\}}}{\qdown_{S \setminus \{a,s_1,\ldots,s_{i-1}\}}}
 ~\leq~ \prod_{i=1}^{k} \frac{\qdown_{T \setminus \{a,s_1,\ldots,s_{i-1},s_i\}}}{\qdown_{T \setminus \{a,s_1,\ldots,s_{i-1}\}}}
  ~=~ \frac{\qdown_{T \setminus (S \cap \Gamma^+(a))}}{\qdown_{T \setminus \{a\}} } 
 ~\leq~ \frac{\qdown_{T \setminus \Gamma^+(a)}}{\qdown_{T \setminus \{a\}} } $$
where in the last step we used the monotonicity of $\qdown_T$ in $T$ (again from \Claim{breve1}). This implies \Equation{breve-ratios}:
$$ \frac{\qdown_{S}}{\qdown_{S \setminus \{a\}}} ~=~ 1 - p_a \frac{\qdown_{S \setminus \Gamma^+(a)}}{\qdown_{S \setminus \{a\}}}
 ~\geq~ 1 - p_a \frac{\qdown_{T \setminus \Gamma^+(a)}}{\qdown_{T \setminus \{a\}}}
 ~=~ \frac{\qdown_{T}}{\qdown_{T \setminus \{a\}}}. $$
\end{proof}

\begin{claim}[Log-submodularity of $q_S$]
\ClaimName{submodular}
For any $p \in \cS$ and $A,B \subseteq [n]$, we have
$q_A \cdot q_B \geq q_{A \union B} \cdot q_{A \intersect B}$.
\end{claim}

\begin{proof}
We can assume $A \cup B \in \Ind$; otherwise the right-hand side is zero. By \Claim{breve3}, we have
$$ q_{A} \cdot  q_{B} ~=~ p^{A} \qdown_{[n] \setminus \Gamma^+(A)}  \cdot p^{B} \qdown_{[n] \setminus \Gamma^+(B)}. $$
By \Claim{breve-submodular},
$$\qdown_{[n] \setminus \Gamma^+(A)} \cdot \qdown_{[n] \setminus \Gamma^+(B)} 
 ~\geq~ \qdown_{[n] \setminus (\Gamma^+(A) \cup \Gamma^+(B))} \cdot \qdown_{[n] \setminus
 (\Gamma^+(A) \cap \Gamma^+(B))}.$$
Here we use the fact that $\Gamma^+(A) \cup \Gamma^+(B) = \Gamma^+(A \cup B)$, and $\Gamma^+(A) \cap \Gamma^+(B) \supseteq \Gamma^+(A \cap B)$. Therefore, by the monotonicity of $\qdown_S$,
$$\qdown_{[n] \setminus \Gamma^+(A)} \cdot \qdown_{[n] \setminus \Gamma^+(B)} 
 ~\geq~ \qdown_{[n] \setminus \Gamma^+(A \cup B)}
   \cdot \qdown_{[n] \setminus \Gamma^+(A \cap B)}.$$
Also, $p^A p^B = p^{A \cup B} p^{A \cap B}$.
Using \Claim{breve3} one more time, we obtain
$$ q_{A} \cdot  q_{B} 
 ~\geq~ p^{A \cup B} \qdown_{[n] \setminus \Gamma^+(A \cup B)}  \cdot p^{A \cap B} \qdown_{[n] \setminus \Gamma^+(A \cap B)}
 ~=~ q_{A \cup B} \cdot q_{A \cap B}.$$
\end{proof}

\begin{claim}
\ClaimName{sumqJ}
Suppose that $p \in \cS$.
For any set $S \subseteq [n]$,
$$
\sum_{J \subseteq S} \frac{q_J}{q_\emptyset}
    ~\leq~ \prod_{j \in S} \Big( 1 + \frac{q_{\set{j}}}{q_\emptyset} \Big).
$$
\end{claim}

\begin{proof}
The proof is by induction on $S$, the case $\card{S} \leq 1$ being trivial.
Fix any $s \in S$.
\Claim{submodular} implies that $q_{J+s} \cdot q_\emptyset \leq q_{\set{s}} \cdot q_J$
for any $J \subseteq S \setminus \{s\}$.
Summing over $J$ yields
$$
\sum_{J \subseteq S \setminus \{s\}} \frac{q_{J+s}}{q_\emptyset}
    ~\leq~ \frac{q_{\set{s}}}{q_\emptyset} \sum_{J \subseteq S \setminus \{s\}} \frac{q_J}{q_\emptyset}.
$$
Adding $\sum_{J \subseteq S \setminus \{s\}} \frac{q_{J}}{q_\emptyset}$ to both sides yields
$$
\sum_{J \subseteq S} \frac{q_{J}}{q_\emptyset}
    ~\leq~ \Big(1 + \frac{q_{\set{s}}}{q_\emptyset}\Big)
           \sum_{J \subseteq S \setminus \{s\}} \frac{q_J}{q_\emptyset}.
$$
The claim follows by induction.
\end{proof}

\begin{claim}
\ClaimName{qSingleton}
If $q_\emptyset > 0$ then 
$\frac{q_{\{i\}}}{q_\emptyset} = \frac{\qdown_{[n] \setminus \{i\}}}{\qdown_{[n]}} - 1$.
\end{claim}

\begin{proof}
By \Claim{breve2},
$$ 1+\frac{q_{\{i\}}}{q_\emptyset} = \frac{q_\emptyset + q_{\{i\}}}{q_\emptyset}
 = \frac{\qdown_{[n] \setminus \{i\}}}{\qdown_{[n]}}.$$
\end{proof}

\begin{claim}[{\protect Kolipaka-Szegedy \cite[Theorem 5]{Kolipaka}}]
\ClaimName{ShearerSlack}
If $(1+\epsilon) p \in \cS$ then
$\frac{q_{\{i\}}}{q_\emptyset} \leq \frac{1}{\epsilon}$ for each $i \in [n]$.
\end{claim}

\begin{proof}
Note that $\qdown_{[n] \setminus \{i\}}(p)$ does not depend on $p_i$, while $\qdown_{[n]}(p)$ is
linear in $p_i$. Also, both quantities are equal at $p_i=0$:
we have $\qdown_{[n]}(p_1,\ldots, 0 \cdot p_i, \ldots, p_n) = \qdown_{[n] \setminus \{i\}}(p)$.
Since $(1+\epsilon) p \in \cS$, we know that
$\qdown_{[n]}(p_1,\ldots, (1+\epsilon) p_i, \ldots, p_n) \geq 0$. By linearity, $\qdown_{[n]}(p)
\geq \frac{\epsilon}{1+\epsilon} \qdown_{[n] \setminus \{i\}}(p)$.
\Claim{qSingleton} then implies that $\frac{q_{\{i\}}}{q_\emptyset} \leq \frac{1}{\epsilon}$.
\end{proof}

\subsubsection{Connection to stable set sequences}

Kolipaka and Szegedy showed that stable set sequences relate to the independence polynomials $q_S$.
The following is the crucial upper-bound for stable set sequences when
Shearer's criterion holds.
In fact, this result is subsumed by \Lemma{Shearer-sum-equality} but we present
the upper bound first, with a shorter proof. 

\begin{lemma}[Kolipaka-Szegedy~{\protect \cite{Kolipaka}}]
\LemmaName{Shearer-sum-bound}
If $q_S \geq 0$ for all $S \subseteq [n]$ and $q_\emptyset > 0$, then
$$ \sum_{\cI \in \Stab_\ell(J)} \!\!\!\! p_\cI ~\leq~ \frac{q_{J}}{q_\emptyset}
\qquad\forall J \in \Ind, \forall \ell \geq 1.
$$
\end{lemma}

\begin{proof}
We proceed by induction: for $\ell=1$, there is only one such stable set sequence $\cI = (J)$.
By Lemma~\ref{lem:q-expansion}, we have
$q_J = p^J \sum_{S \subseteq \Gamma^+(J)} q_S \geq p^J q_\emptyset$.
(Recall that $q_S \geq 0$ for all $S \subseteq [n]$.)
Hence, $p_{(J)} = p^J \leq q_J / q_\emptyset$.

The inductive step: every stable set sequence starting with $J$ has the form $\cI = (J,J',\ldots)$ where $J' \subseteq \Gamma^+(J)$. Therefore,
\begin{equation}
\label{eq:stab-rec}
 \sum_{\cI \in \Stab_\ell(J)} p_\cI
~=~ p^J \sumstack{J' \in \Ind \\ J' \subseteq \Gamma^+(J)}
      ~ \sum_{\cI \in \Stab_{\ell-1}(J')} \!\!\!\! p_\cI.
\end{equation}
By the inductive hypothesis, $\sum_{\cI \in \Stab_{\ell-1}(J')} p_\cI \leq {q_{J'}}/{q_\emptyset}$. Also, recall that $q_{J'}=0$ if $J' \notin \Ind$. Therefore,
$$ \sum_{\cI \in \Stab_\ell(J)} p_\cI  \leq p^J \sum_{J' \subseteq \Gamma^+(J)} \frac{q_{J'}}{q_\emptyset}
 = \frac{q_J}{q_\emptyset} $$
using Lemma~\ref{lem:q-expansion} to obtain the last equality.
\end{proof}

The inequality in \Lemma{Shearer-sum-bound} actually becomes an equality as
$\ell \rightarrow \infty$, as shown in \Lemma{Shearer-sum-equality}.
This stronger result is used only tangentially in \Section{CLLSSS}, but we provide a detailed proof
in order to clarify the arguments of Kolipaka and Szegedy \cite{Kolipaka}.

\begin{lemma}[Kolipaka-Szegedy~{\protect \cite[Theorem 14]{Kolipaka}}]
\label{lem:Shearer-sum-equality}
For a dependency graph $G$ and $p_1,\ldots,p_n \in (0,1)$, 
the following statements are equivalent:
\begin{enumerate}
\item $q_\emptyset > 0$ and $q_S \geq 0$ for all $S \subseteq [n]$.
\item for all $J \in \Ind$, $q_J > 0$ and $ \sum_{\cI \in \Prop(J)} p_\cI ~=~ q_{J}/q_\emptyset$.
\item $\sum_{\cI \in \Prop(J)} p_\cI$ is finite for each $J \in \Ind$.
\end{enumerate}
\end{lemma}

\begin{proof}
First, note that $\Prop(J) = \bigcup_{t=1}^{\infty} \Prop_t(J)$, and $\bigcup_{t=1}^{\ell} \Prop_t(J)$ can be identified with $\Stab_\ell(J)$, since each proper sequence $\cI$ of length at most $\ell$ can be padded with empty sets to obtain a sequence in $\Stab_\ell(J)$ (and $p_\cI$ does not change). Therefore, $\sum_{\cI \in \Prop(J)} p_\cI = \lim_{\ell \rightarrow \infty} \sum_{\cI \in \Stab_\ell(J)} p_\cI$. 
This is a non-decreasing sequence; the limit exists but could be infinite.  Let us denote $w^{(\ell)}_J= \sum_{\cI \in \Stab_\ell(J)} p_\cI$ and $w^*_J =  \lim_{\ell \rightarrow \infty} w^{(\ell)}_J = \sum_{\cI \in \Prop(J)} p_\cI$.
Let us define $M$ to be the following linear operator on $\RR^\Ind$:
$$ (Mx)_I = p^I \sumstack{J \in \Ind \\ J \subseteq \Gamma^+(I)} x_J.$$
Using this notation, the identity \eqref{eq:stab-rec} can written compactly as $w^{(\ell)} = M w^{(\ell-1)}$.
Inductively, $w^{(\ell)} = M^{\ell-1} w^{(1)}$, and $w^* = \lim_{\ell \rightarrow \infty} M^{\ell} w^{(1)}$. 

{$\mathbf 1 \Rightarrow \mathbf 2$:}
Assume now that $q_S \geq 0$ for all $S \subseteq [n]$ and $q_\emptyset > 0$.
Lemma~\ref{lem:Shearer-sum-bound} proves that this implies $w^*_J = \sum_{\cI \in \Prop(J)} p_\cI =
\lim_{\ell \rightarrow \infty} \sum_{\cI \in \Stab_\ell(J)} p_\cI \leq q_J / q_\emptyset$.
Clearly $ \sum_{\cI \in \Prop(J)} p_\cI > 0$, so this also implies that $q_J > 0$ for all $J \in \Ind$.

Note that $w^{(1)}$ is the column of $M$ corresponding to $J=\emptyset$: $M_{I,\emptyset} = p^I$ for
each $I \in \Ind$. Therefore, we can write $w^{(1)} = M w^{(0)}$, where $w^{(0)} = e_\emptyset$ is
the canonical basis vector in $\RR^\Ind$ corresponding to $\emptyset$. We have $w^* = \lim_{\ell \rightarrow \infty} M^\ell w^{(1)} = \lim_{\ell \rightarrow \infty} M^\ell w^{(0)}$.
We may subtract these two limits since we have shown that every $w^*_J$ is finite,
obtaining $\lim_{\ell \rightarrow \infty} M^\ell (w^{(1)} - w^{(0)}) = 0$.
We note that $w^{(1)} - w^{(0)}$ has strictly positive coordinates for $I \neq \emptyset$, and $0$ for $I = \emptyset$. 

By Lemma~\ref{lem:q-expansion}, we have $Mq = q$ for the vector $q \in \RR^\Ind$ with coordinates $q_I$. 
Consider $\frac{1}{q_\emptyset} q - w^{(0)}$, a nonnegative vector with $0$ in the coordinate corresponding to $\emptyset$.  We can choose $\beta > 0$ large enough so that coordinate-wise, $0 \leq \frac{1}{q_\emptyset} q - w^{(0)} \leq \beta (w^{(1)} - w^{(0)})$. 
From this we derive that
$$ 0 \leq \frac{1}{q_\emptyset} q - w^* = \lim_{\ell \rightarrow \infty} M^\ell
\left(\frac{1}{q_\emptyset} q - w^{(0)} \right) \leq \beta \lim_{\ell \rightarrow \infty} M^\ell
(w^{(1)} - w^{(0)}) = 0,$$
so equality holds throughout.
Recalling the definition of $w^*_J$, we conclude that
$\sum_{\cI \in \Prop(J)} p_\cI = w^*_J = \frac{1}{q_\emptyset} q_J$.

{$\mathbf 2 \Rightarrow \mathbf 3$:} Trivial.

{$\mathbf 3 \Rightarrow \mathbf 1$:}
Let $p \in (0,1)^n$ be the vector $(p_1,\ldots,p_n)$.
We can assume that $\min_S \qdown_S(p) \leq 0$, otherwise we are done by \Claim{QAndQDown}.
Let us consider the values of $\qdown_S$ on the line $\setst{ \lambda p }{ \lambda \in [0,1] }$.
Define $\lambda^* = \inf \{ \lambda \in (0,1] : \min_S \qdown_S(\lambda p) \leq 0 \}$.
We observe that $\min_S \qdown_S(\lambda p) > 0$ for $0 < \lambda < 1/n$, which can be
verified directly by considering the alternating sum defining $\qdown_S$. (Intuitively,
Shearer's Lemma holds in this region just by the union bound.)
Therefore, we have $\lambda^* > 0$.
Furthermore continuity also implies $\min_S \qdown_S(\lambda^* p) = 0$,
so \Claim{QAndQDown} yields $q_\emptyset(\lambda^* p) = \qdown_{[n]}(\lambda^* p) = 0$.
For $\lambda \in [0,\lambda^*)$ we have $\min_S \qdown_S(\lambda p) > 0$,
so by \Claim{QAndQDown} we also have $\min_{I \in \Ind} q_I(\lambda p) > 0$.
This shows that the condition $\mathbf 1$ holds at the point $\lambda p$,
for $\lambda \in [0,\lambda^*)$,
so we may use the implication $\mathbf 1 \Rightarrow \mathbf 2$:
$ \sum_{\cI \in \Prop(J)} (\lambda p)_\cI = q_J(\lambda p) / q_\emptyset(\lambda p)$.
Let $J \in \Ind$ be such that $q_J(\lambda^* p) > 0$;
such a $J$ must exist by \Claim{q-sum}.
By the monotonicity of $p_\cI = \prod_{I \in \cI} p^I$ in the variables $p_1,\ldots,p_n$, we have 
$$
    \sum_{\cI \in \Prop(J)} p_\cI 
    ~\geq~ \sum_{\cI \in \Prop(J)} (\lambda^* p)_\cI 
    ~\geq~ \liminf_{\lambda \rightarrow \lambda^*-} \sum_{\cI \in \Prop(J)} (\lambda p)_\cI 
    ~=~ \liminf_{\lambda \rightarrow \lambda^*-} \frac{ q_J( \lambda p) }{ q_\emptyset( \lambda p ) }
    ~=~ \infty,
$$
as $q_J(\lambda^* p) > 0$ but $q_\emptyset(\lambda^* p) = 0$.
This contradicts the assumption {\bf 3} that $\sum_{\cI \in \Prop(J)} p_\cI$ is finite.
\end{proof}

From \Claim{q-sum}, we obtain immediately the following.

\begin{corollary}
\label{cor:Shearer-bound}
If $q_S \geq 0$ for all $S \subseteq [n]$ and $q_\emptyset > 0$,
$$\sum_{\cI \in \Prop} p_\cI = \frac{1}{q_\emptyset}.$$
\end{corollary}

\noindent
\textbf{Remark.} An equivalent statement using the language of ``traces'' appears in the recent manuscript of Knuth
\cite[Page 86, Theorem F]{Knuth}, together with a short proof using generating functions.
Furthermore, using \Claim{breve2}, we may derive
$$
\sum_{J \subseteq A} \: \sum_{\cI \in \Prop(J)} p_\cI
    ~=~ \sum_{J \subseteq A} \frac{q_{J}}{q_\emptyset}
    ~=~ \frac{\qdown_{[n] \setminus A}}{\qdown_{[n]}},
$$
for any $A \subseteq [n]$.
This statement, in the language of traces,
also appears in Knuth's draft \cite[Page 87, Equation (144)]{Knuth}.

\paragraph{Summary at this point.}
By \Lemma{iteration-bound} and \Corollary{Shearer-bound}, \MSR produces a state in $\bigcap_{i=1}^n \overline{E_i}$
after at most $1 / q_\emptyset$ iterations in expectation.
However, this should not be viewed as a statement of efficiency.
Shearer's Lemma proves that
$\Pr_\mu[ \bigcap_{i=1}^n \overline{E_i} ] \geq q_\emptyset$
so, in expectation, $1/q_\emptyset$ independent samples from $\mu$
would also suffice to find a state in $\bigcap_{i=1}^n \overline{E_i}$.

\Section{Shearerslack} improves this analysis by assuming that Shearer's criterion holds with some
slack, analogous to the result in \Section{LLLslack}.
\Section{shearer-automatic-slack} then removes the need for that assumption ---
it argues that Shearer's criterion always holds with some slack,
and provides quantitative bounds on that slack.

\subsection{Shearer's criterion with slack}
\SectionName{Shearerslack}
\SectionName{shearer-slack}

In this section we consider scenarios in which Shearer's criterion holds with a certain
amount of slack.
To make this formal, we will consider another vector $p'$ of probabilities with $p \leq p' \in \cS$.
For notational convenience, we will let $q'_S$ denote the value $q_S(p')$
and let $q_S$ denote $q_S(p)$ as before.
Let us assume that Shearer's criterion holds with some slack in the following natural sense.

\begin{definition}
\DefinitionName{ShearerSlack}
We say that $p \in (0,1)^n$ satisfies Shearer's criterion with coefficients $q'_S$
at a slack of $\epsilon$, if $p' = (1+\epsilon) p$ is still in the Shearer region $\cS$ and $q'_S = q_S(p')$.
\end{definition}

\begin{theorem}
\TheoremName{Shearer-slack}
Recall that $p_i = \Pr_\mu[E_i]$.
If the $p_i$ satisfy Shearer's criterion with coefficient $q'_\emptyset$ at a slack of $\epsilon \in (0,1)$,
then the probability that \MSR resamples more than
$\frac{2}{\epsilon} \big(\ln \frac{1}{q'_\emptyset} + t\big)$ events is at most $e^{-t}$.
\end{theorem}

\begin{proof}
By Lemma~\ref{lem:iteration-bound}, the probability that \MSR resamples more than $s$ events is at most $\sum_{\cI \in \Prop: \sigma(\cI)=\lceil s \rceil} p_\cI$. By the slack assumption, we have
$$ \Pr[\mbox{resample more than }s\mbox{ events}]
 ~\leq~ \sumstack{\cI \in \Prop \\ \sigma(\cI)=\lceil s \rceil} p_\cI
 ~\leq~ (1+\epsilon)^{-s} \sumstack{\cI \in \Prop \\ \sigma(\cI)= \lceil s \rceil} p'_\cI
$$
since we have $p'_i = (1+\epsilon) p_i$ for each event appearing in a sequence $\cI$. The
hypothesis is that the probabilities $p'_i$ satisfy Shearer's criterion with a bound of $q'_\emptyset$.
Consequently, \Corollary{Shearer-bound} implies that
$\sum_{\cI \in \Prop: \sigma(\cI)= \lceil s \rceil} p'_\cI \leq \sum_{\cI \in \Prop} p'_\cI \leq 1 / q'_\emptyset$.
Thus, for $s = \frac{2}{\epsilon} \big(\ln \frac{1}{q'_\emptyset} + t\big)$ we obtain
$$ \Pr[\mbox{resample more than }s\mbox{ events}] ~\leq~ (1+\epsilon)^{-s} \frac{1}{q'_\emptyset}
~\leq~ e^{-s \epsilon/2} \frac{1}{q'_\emptyset}
~\leq~ e^{-(\ln ({1}/{q'_\emptyset}) + t)} \frac{1}{q'_\emptyset}
~=~ e^{-t}.$$
\end{proof}

In other words, the probability that \MSR requires more than
$\frac{2}{\epsilon} \ln(1/q'_\emptyset)$ resamplings decays exponentially fast; in particular the
expected number of resampled events is $O\big(\frac{1}{\epsilon} \ln(1/q'_\emptyset)\big)$.
This appears significantly better than the trivial bound of $1/q_\emptyset$; still, it is not clear whether this bound can be considered ``polynomial". In the following, we show that this leads in fact to efficient bounds, comparable to the best known bounds in the variable model.

\begin{corollary}
\CorollaryName{Shearer-no-q0}
If the $p_i$ satisfy Shearer's criterion with coefficients $q'_S$
at a slack of $\epsilon \in (0,1)$, then the probability that \MSR resamples more than
$$
\frac{2}{\epsilon}\Bigg( \sum_{j=1}^n \ln \Big( 1 + \frac{q'_{\set{j}}}{q'_\emptyset} \Big) + t\Bigg)
$$
events is at most $e^{-t}$.
\end{corollary}
\begin{proof}
By \Claim{q-sum} and \Claim{sumqJ}, we have
\begin{align*}
\ln \frac{1}{q'_\emptyset}
 ~=~ \ln \sum_{J \subseteq [n]} \frac{q'_J}{q'_\emptyset} 
 ~\leq~ \sum_{j=1}^{n} \ln \Big( 1 + \frac{q'_{\set{j}}}{q'_\emptyset} \Big).
\end{align*}
The result follows from \Theorem{Shearer-slack}.
\end{proof}

Next, we provide a simplified bound that depends only on the amount of slack and the number of events. This is analogous to a bound of $O(n / \epsilon)$ given by Kolipaka-Szegedy \cite{Kolipaka} in the variable model.

\begin{theorem}
\TheoremName{Shearer-slack-simple}
If $p_1,\ldots,p_n$ satisfy Shearer's criterion at a slack of $\epsilon \in (0,1)$, then the expected number of events resampled by \MSR is $O(\frac{n}{\epsilon} \log \frac{1}{\epsilon})$.
\end{theorem}

\begin{proof}
Let $p' = (1+\epsilon/2) p$. By assumption,
$(1+\epsilon/3) p' \leq (1+\epsilon) p \in \cS$. Therefore, $p'$ still has $\epsilon/3$ slack so by
\Claim{ShearerSlack}, the coefficients $q'_S = q_S(p')$ satisfy
$\frac{q'_{\{i\}}}{q'_\emptyset} \leq \frac{3}{\epsilon}$.
The point $p$ satisfies Shearer's criterion with coefficients $q'_S$ at a slack of $\epsilon/2$, so by \Corollary{Shearer-no-q0}, the probability that we resample more than $\frac{4}{\epsilon} (n \ln (1 + \frac{3}{\epsilon}) + t)$ events is at most $e^{-t}$. In expectation, we resample $O(\frac{n}{\epsilon} \log \frac{1}{\epsilon})$ events as claimed.
\end{proof}

\subsection{Quantification of slack in Shearer's criterion}
\SectionName{shearer-automatic-slack}

In the previous section, we proved a bound on the number of resamplings in the \MSR algorithm, provided that Shearer's criterion is satisfied with a certain slack. In fact, from \Definition{Shearer-region} one can observe that the Shearer region is an {\em open set} and therefore there is always a certain amount of slack.
However, how large a slack we can assume is not a priori clear. In particular, one can compare with
Kolipaka-Szegedy \cite{Kolipaka} where a bound is proved on the expected number of events one has to
resample in the variable model: If Shearer's criterion is satisfied with coefficients $q_S$, then
the expected number of resamplings is at most $\sum_{i=1}^{n} q_{\{i\}} / q_\emptyset$
\cite{Kolipaka}. In this section, we prove that anywhere in the Shearer region, there is an amount
of slack {\em inversely proportional to this quantity}, which leads to a bound similar to that of
Kolipaka and Szegedy \cite{Kolipaka}. 

\begin{lemma}
\LemmaName{Shearer-automatic-slack}
Let $(p_1,\ldots,p_n) \in (0,1)^n$ be a point in the Shearer region.
Let $\epsilon = q_\emptyset / (2\sum_{i=1}^{n} q_{\{i\}})$ and $p'_i = (1+\epsilon) p_i$.
Then $(p'_1,\ldots,p'_n)$ is also in the Shearer region, and 
$q_\emptyset(p') \geq \frac12 q_\emptyset(p)$.
\end{lemma}

Before proving the lemma, let us consider the partial derivatives of the $\qdown_S$ polynomials.

\begin{claim}
\ClaimName{breve-diff}
For any $i \in S$, 
$$\partdiff{\qdown_S}{p_i} ~=~ -\qdown_{S \setminus \Gamma^+(i)}$$
and for any $j \in S \setminus \Gamma^+(i)$,
$$\mixdiff{\qdown_S}{p_i}{p_j} ~=~ \qdown_{S \setminus \Gamma^+(i) \setminus \Gamma^+(j)}.$$
For other choices of $i,j$, the partial derivatives are $0$.
In particular, for any point in the Shearer region, $\partdiff{\qdown_S}{p_i} \leq 0$ and
$\mixdiff{\qdown_S}{p_i}{p_j} \geq 0$.
\end{claim}

Due to \Claim{breve-diff}, we may say that $\qdown_S(p_1,\ldots,p_n)$ is ``continuous supermodular"
in the Shearer region. 

\begin{proof}
For any $i \in S$, we have $\qdown_S = \qdown_{S \setminus \{i\}} - p_i \qdown_{S \setminus
\Gamma^+(i)}$ by \Claim{breve1}. The polynomials $\qdown_{S \setminus \{i\}}$ and $\qdown_{S \setminus \Gamma^+(i)}$ do not depend on $p_i$ and hence $\partdiff{\qdown_S}{p_i}$ is equal to $-\qdown_{S \setminus \Gamma^+(i)}$. Repeating this argument one more time for $j \in S \setminus \Gamma^+(i)$, we get $\partdiff{\qdown_S}{p_i} = -\qdown_{S \setminus \Gamma^+(i)} = -\qdown_{S \setminus \Gamma^+(i) \setminus \{j\}} + p_j \qdown_{S \setminus \Gamma^+(i) \setminus \Gamma^+(j)}$. Again, $\qdown_{S \setminus \Gamma^+(i) \setminus \{j\}}$ and $\qdown_{S \setminus \Gamma^+(i) \setminus \Gamma^+(j)}$ do not depend on $p_j$ and hence $\mixdiff{\qdown_S}{p_i}{p_j} = \qdown_{S \setminus \Gamma^+(i) \setminus \Gamma^+(j)}$.

Clearly, we have $\partdiff{\qdown_S}{p_i} = 0$ unless $i \in S$, and $\mixdiff{\qdown_S}{p_i}{p_j} = 0$ unless $i \in S$ and $j \in S \setminus \Gamma^+(i)$. Since all the coefficients $\qdown_S$ are positive in the Shearer region, we have $\partdiff{\qdown_S}{p_i} \leq 0$ and $\mixdiff{\qdown_S}{p_i}{p_j} \geq 0$ for all $i,j$.
\end{proof}

Now we can prove \Lemma{Shearer-automatic-slack}.

\begin{proof}
Consider the line segment from $p = (p_1,\ldots,p_n)$ to $p' = (p'_1,\ldots,p'_n)$ where $p'_i =
(1+\epsilon) p_i$, $\epsilon = \frac{q_\emptyset}{2\sum_{i=1}^{n} q_{\{i\}}}$. Note that $p'_i \leq
(1 + \frac{q_\emptyset}{q_{\{i\}}}) p_i = \frac{q_{\{i\}} + q_\emptyset}{q_{\{i\}}} p_i =
\frac{\qdown_{[n] \setminus \{i\}}}{p_i \qdown_{[n] \setminus \Gamma^+(i)}} p_i \leq 1$ by
\Claim{breve2}, \Claim{breve3} and \Claim{brevemonotone}.
Let us define
$$ Q_\emptyset(\lambda) ~=~ q_\emptyset((1+\lambda) p_1, \ldots, (1+\lambda) p_n).$$
By the chain rule and \Claim{breve-diff}, we have
$$ \frac{d Q_\emptyset}{d \lambda} \Big|_{\lambda=0} ~=~ \sum_{i=1}^{n} p_i
\partdiff{q_\emptyset}{p_i} ~=~ -\sum_{i=1}^{n} p_i \qdown_{[n] \setminus \Gamma^+(i)} ~=~ -\sum_{i=1}^{n} q_{\{i\}} $$
where we used \Claim{breve3} in the last equality. 
Assuming that $(1+\lambda) p = ((1+\lambda) p_1, \ldots, (1+\lambda) p_n)$ is in the Shearer region, we also have by \Claim{breve-diff}
$$ \frac{d^2 Q_\emptyset}{d \lambda^2} ~=~ \sum_{i,j=1}^{n} \mixdiff{q_\emptyset}{p_i}{p_j} p_i p_j \geq 0.$$
That is, $Q_\emptyset(\lambda)$ is a convex function for $\lambda \geq 0$ as long as $(1+\lambda) p$ is in the Shearer region.
Our goal is to prove that this indeed happens for $\lambda \in [0,\epsilon]$.

Assume for the sake of contradiction that $(1+\lambda) p$ is not in the Shearer region for some $\lambda \in [0,\epsilon]$, and let $\lambda^*$ be the minimum such value (which exists since the complement of the Shearer region is closed). By \Claim{QAndQDown}, anywhere in the Shearer region, $q_\emptyset = \qdown_{[n]}$ is the minimum of the $\qdown_S$ coefficients; hence by continuity it must be the case that $\qdown_{[n]}((1+\lambda^*)p)$ is the minimum coefficient among $\qdown_S((1+\lambda^*)p)$ for all $S \subseteq [n]$, and $Q_\emptyset(\lambda^*) = \qdown_{[n]}((1+\lambda^*)p) \leq 0$.
On the other hand, by the minimality of $\lambda^*$, $Q_\emptyset(\lambda)$ is positive and convex on $[0,\lambda^*)$ and therefore
$$ Q_\emptyset(\lambda^*) ~\geq~ Q_\emptyset(0) + \lambda^* \frac{d Q_\emptyset}{d \lambda}
\Big|_{\lambda=0} ~=~ q_\emptyset - \lambda^* \sum_{i=1}^{n} q_{\{i\}} ~\geq~ q_\emptyset - \epsilon
\sum_{i=1}^{n} q_{\{i\}} ~=~ \frac12 q_\emptyset ~>~ 0, $$
which is a contradiction. Therefore, $Q_\emptyset(\lambda)$ is positive and convex for all $\lambda
\in [0,\epsilon]$. By the same computation as above, $Q_\emptyset(\epsilon) \geq \frac12 q_\emptyset$.
\end{proof}

This implies our main algorithmic result under Shearer's criterion.

\begin{theorem}
\TheoremName{Shearer-no-slack}
Let $E_1,\ldots,E_n$ be events and let $p_i = \Pr_\mu[E_i]$.
Suppose that the three subroutines described in \Section{algass} exist.
If $p \in \cS$ then the probability that \MSR resamples more than
$4 \sum_{i=1}^{n} \frac{q_{\set{i}}}{q_\emptyset} \big( \sum_{j=1}^{n} \ln (1 + \frac{q_{\{j\}}}{q_\emptyset}) + 1 + t\big)$ events is at most $e^{-t}$.
\end{theorem}

We note that the corresponding result in the variable model \cite{Kolipaka}
was that the expected number of resamplings is at most $\sum_{i=1}^{n} \frac{q_{\{i\}}}{q_\emptyset}$.
Here, we obtain a bound which is at most quadratic in this quantity. 

\begin{proof}
Directly from \Theorem{Shearer-slack} and \Lemma{Shearer-automatic-slack}: Given $p$ in the Shearer
region, \Lemma{Shearer-automatic-slack} implies that $p$ in fact satisfies Shearer's criterion with
a bound of $q'_\emptyset \geq \frac{q_\emptyset}{2}$ at a slack of
$\epsilon = \frac{q_\emptyset}{2} / \sum_{i=1}^{n} q_{\set{i}}$.
By \Theorem{Shearer-slack}, the probability that \MSR resamples more than $s$ events is at most $e^{-t}$, where
$$ s ~=~ \frac{2}{\epsilon} \left(\ln \frac{1}{q'_\emptyset} + t \right) ~\leq~ \frac{4}{q_\emptyset}  \sum_{i=1}^{n} q_{\{i\}} \left(\ln \frac{1}{q_\emptyset} + 1 + t \right).$$
Using \Claim{sumqJ}, we can replace $\ln \frac{1}{q_\emptyset}$ by $\sum_{j=1}^{n} \ln (1 + \frac{q_{\{j\}}}{q_\emptyset})$.
\end{proof}

\subsection{The General LLL criterion, without slack}
\SectionName{LLLimpliesShearer}

Shearer's Lemma (\Lemma{Shearer-Lemma}) is a strengthening of the
original Lov\'asz Local Lemma (\Theorem{LLL}):
if $p_1,\ldots,p_n$ satisfy \eqref{eq:GLL} then they must also satisfy Shearer's criterion
$p \in \cS$.
Nevertheless, there does not seem to be a direct proof of this fact in the literature.
Shearer \cite{Shearer} indirectly proves this fact by showing that, when $p \not\in \cS$
it is possible that $\Pr[ \bigcap_{i=1}^n \overline{E_i} ] = 0$,
so the contrapositive of \Theorem{LLL} implies that \eqref{eq:GLL} cannot hold.
Scott and Sokal prove this fact using analytic properties of the
partition function \cite[Corollary 5.3]{ScottSokal}.
In this section we establish this fact by an elementary, self-contained proof.

We then establish \Theorem{LLL-tight-result},
our algorithmic form of \Theorem{LLL} in the general framework of resampling oracles.
Unlike the simpler analysis of \Section{LLLslack},
the analysis of this section does not assume any slack in the \eqref{eq:GLL} criterion.

\begin{lemma}
\LemmaName{GLLimpliesShearer}
Suppose that $p$ satisfies \eqref{eq:GLL}.
Then, for every $S \subseteq [n]$ and $a \in S$, we have
$$ \frac{\qdown_S}{\qdown_{S \setminus \set{a}}} ~\geq~ 1-x_a .$$
\end{lemma}

\begin{corollary}[\eqref{eq:GLL} implies Shearer]
\CorollaryName{Shearer->Lovasz}
If $p$ satisfies \eqref{eq:GLL} then $p \in \cS$.
\end{corollary}
\begin{proof}
For any $S \subseteq [n]$, write it as $S = \set{s_1,\ldots,s_k}$. Induction yields
$$
\frac{\qdown_S}{\qdown_\emptyset}
 ~=~ \prod_{i=1}^{k} \frac{\qdown_{\set{s_1,\ldots,s_i}}}{\qdown_{\set{s_1,\ldots,s_{i-1}}}}
 ~\geq~ \prod_{a \in S} (1-x_a)
 ~>~ 0.
$$
The claim follows since $\qdown_\emptyset=1$.
\end{proof}

\begin{corollary}
\CorollaryName{qAndx}
If $p$ satisfies \eqref{eq:GLL} then $\frac{q_{\set{a}}}{q_\emptyset} \leq \frac{x_a}{1-x_a}$.
\end{corollary}
\begin{proof}
\Lemma{GLLimpliesShearer} yields
$\frac{\qdown_{[n] - a}}{\qdown_{[n]}} \leq \frac{1}{1-x_a}$,
so the result follows from \Claim{ShearerSlack}.
\end{proof}

\begin{proofof}{\Lemma{GLLimpliesShearer}}
We proceed by induction on $|S|$.
The base case, $S = \emptyset$, is trivial: there is no $a \in S$ to choose. Consider $S \neq
\emptyset$ and an element $a \in S$. By \Claim{breve1}, we have $\qdown_S = \qdown_{S
\setminus \{a\}} - p_a \qdown_{S \setminus \Gamma^+(a)}$. By the inductive hypothesis applied
iteratively to the elements of
$(S \setminus \set{a}) \setminus (S \setminus \Gamma^+(a)) = \Gamma(a) \cap S$, we have
$$ \qdown_{S \setminus \{a\}}
   ~\geq~
    \qdown_{S \setminus \Gamma^+(a)} \prod_{i \in \Gamma(a) \cap S} (1-x_i).$$
Therefore, we can write
$$
\qdown_S
    ~=~ \qdown_{S \setminus \{a\}} - p_a \qdown_{S \setminus \Gamma^+(a)}
    ~\geq~ \qdown_{S \setminus \{a\}} \left(1 - \frac{p_a}{ \prod_{i \in \Gamma(a) \cap S} (1-x_i)} \right).
$$
By the claim's hypothesis,
$p_a \leq x_a \prod_{i \in \Gamma(a)} (1-x_i) \leq x_a \prod_{i \in \Gamma(a) \cap S} (1-x_i)$,
so we conclude that $\qdown_S \geq (1-x_a) \qdown_{S \setminus \set{a}}$.
\end{proofof}

These results, together with our analysis of Shearer's criterion with slack
(\Corollary{Shearer-no-q0}), 
immediately provide an analysis under the assumption that \eqref{eq:GLL} holds with slack,
similar to \Theorem{GLL-with-slack}.
However, this connection to Shearer's criterion allows us to prove more.

We show that our algorithm is in fact efficient even when the \eqref{eq:GLL} criterion is tight.
This might be surprising in light of \Corollary{Shearer-bound}, which does not use
any slack and gives an exponential bound of
$\frac{1}{q_\emptyset} = \frac{1}{\qdown_{[n]}} \leq \prod_{i=1}^{n} \frac{1}{1-x_i}$.
The reason why we can prove a stronger bound is that Shearer's criterion is {\em never tight}:
as we argued already, it defines an open set, and
\Section{shearer-automatic-slack} derives a quantitative bound on the slack that
is always available under Shearer's criterion.

\begin{theorem}
\TheoremName{Lovasz-no-slack}
Let $E_1,\ldots,E_n$ be events and let $p_i = \Pr_\mu[E_i]$.
Suppose that the three subroutines described in \Section{algass} exist.
If $p$ satisfies \eqref{eq:GLL} then the probability that \MSR resamples more than $4 \sum_{i=1}^{n} \frac{x_i}{1-x_i} (\sum_{j=1}^{n} \ln \frac{1}{1-x_j} + 1 + t)$ events is at most $e^{-t}$.

If \eqref{eq:GLL} is satisfied with a slack of $\epsilon \in (0,1)$,
i.e., $(1+\epsilon) p_i \leq x_i \prod_{j \in \Gamma(i)} (1-x_j)$,
then with probability at least $1-e^{-t}$, \MSR resamples no more than
$ \frac{2}{\epsilon} (\sum_{j=1}^{n} \ln \frac{1}{1-x_j} + t)$ {events}.
\end{theorem}

\begin{proof}
The first part follows directly from \Theorem{Shearer-no-slack},
since \Corollary{Shearer->Lovasz} shows that $p \in \cS$
and \Corollary{qAndx} shows that $\frac{q_{\set{i}}}{q_\emptyset} \leq \frac{x_i}{1-x_i}$.
The second part follows from \Corollary{Shearer-no-q0},
using again that $\frac{q_{\set{i}}}{q_\emptyset} \leq \frac{x_i}{1-x_i}$.
\end{proof}

\Theorem{LLL-tight-result} follows immediately from \Theorem{Lovasz-no-slack}.

\subsection{The cluster expansion criterion}
\SectionName{clusanal}
\SectionName{cluster}

Recall that \Section{generalizingLLL} introduced the cluster expansion criterion,
which often gives improved quantitative bounds compared to the General LLL (such as the applications discussed in \Section{applications}).
For convenience, let us restate the cluster expansion criterion here.
Given parameters $y_1,\ldots,y_n$, 
define the notation
$$
    Y_S = \sumstack{I \subseteq S \\ I \in \Ind} y^I
    \qquad\forall S \subseteq [n].
$$
The cluster expansion criterion for a vector $p \in [0,1]^n$,
with respect to a graph $G$, is
\begin{equation}
\tag{CLL}
\EquationName{CLL}
\exists y_1,\ldots,y_n>0
\qquad\text{such that}\qquad
p_i ~\leq~ y_i / Y_{\Gamma^+(i)}.
\end{equation}
This criterion was introduced in the following non-constructive form of the LLL.

\begin{theorem}[Bissacot et al.~\cite{Bissacot}]
\TheoremName{bissacot}
Let $E_1,\ldots,E_n$ be events with a (lopsi-)dependency graph $G$, and let $p_i = \Pr_\mu[E_i]$.
If $p$ and $G$ satisfy \Equation{CLL} then $\Pr_\mu[\bigcap_{i=1}^{n} \overline{E_i}] > 0$.
\end{theorem}

To see that this strengthens the original LLL (\Theorem{LLL}),
one may verify that \eqref{eq:GLL} implies \eqref{eq:CLL}:
if $p_i \leq x_i \prod_{j \in \Gamma(i)} (1-x_j)$,
we can take $y_i = \frac{x_i}{1-x_i}$ (so $1-x_i = \frac{1}{1+y_i}$)
and then use the simple bound
$$
\sumstack{I \subseteq \Gamma^+(i) \\ I \in \Ind} y^I
~\leq~ \sum_{I \subseteq \Gamma^+(i)} y^I
~=~ \prod_{j \in \Gamma^+(i)} (1+y_j).
$$
On the other hand, Shearer's Lemma (\Lemma{Shearer-Lemma}) strengthens \Theorem{bissacot},
in the sense that \eqref{eq:CLL} implies $p \in \cS$.
This fact was established by Bissacot et al.~\cite{Bissacot} by analytic methods
that relied on earlier results \cite{Fernandez}.
In this section we establish this fact by a new proof that is elementary and self-contained.

An algorithmic form of \Theorem{bissacot} in the variable model was proven by Pegden~\cite{Pegden}.
In fact, that result is subsumed by the algorithm of Kolipaka and Szegedy in Shearer's setting,
since \eqref{eq:CLL} implies $p \in \cS$.
In this section, we prove a new algorithmic form of \Theorem{bissacot}
in the general framework of \ros.

To begin, we establish the following connection between the $y_i$ parameters
and the $\qdown_S$ polynomials.
For convenience, let us introduce the notation $S^c = [n] \setminus S$,
$S+a = S \cup \{a\}$ and $S-a = S \setminus \set{a}$.

\begin{lemma}
\LemmaName{CLLimpliesShearer}
Suppose that $p$ satisfies \eqref{eq:CLL}.
Then, for every $S \subseteq [n]$ and $a \in S$, we have
$$
\frac{\qdown_{S}}{\qdown_{S-a}} ~\geq~ \frac{Y_{S^c}}{Y_{(S-a)^c}}.
$$
\end{lemma}

The proof is in \Section{CLLimpliesShearer} below.

\begin{corollary}[\eqref{eq:CLL} implies Shearer]
\CorollaryName{pInS}
If $p$ satisfies \eqref{eq:CLL} then $p \in \cS$.
\end{corollary}
\begin{proof}
For any $S \subseteq [n]$, write it as $S = \set{s_1,\ldots,s_k}$. Applying \Lemma{CLLimpliesShearer} repeatedly, we obtain
$$
\frac{\qdown_S}{\qdown_\emptyset}
 ~=~ \prod_{i=1}^{k} \frac{\qdown_{\set{s_1,\ldots,s_i}}}{\qdown_{\set{s_1,\ldots,s_{i-1}}}}
 ~\geq~ \prod_{i=1}^{k} \frac{Y_{\set{s_1,\ldots,s_i}^c}}{Y_{\set{s_1,\ldots,s_{i-1}}^c}}
 ~=~ \frac{Y_{S^c}}{Y_{[n]}}
 ~>~ 0
$$
since $Y_T > 0$ for all $T \subseteq [n]$ under the \eqref{eq:CLL} criterion.
Recall that  $\qdown_\emptyset=1$. Hence $\breve{q}_S > 0$ for all $S \subseteq [n]$,
which means that $p$ is in the Shearer region.
\end{proof}

\begin{corollary}
\CorollaryName{qAndy}
If $p$ satisfies \eqref{eq:CLL} then $\frac{q_{\set{a}}}{q_\emptyset} \leq y_a$.
\end{corollary}
\begin{proof}
\Lemma{CLLimpliesShearer} yields 
$\frac{\qdown_{[n] - a}}{\qdown_{[n]}} \leq \frac{Y_{([n]-a)^c}}{Y_{[n]^c}} = 1 + y_a$,
so the result follows from \Claim{ShearerSlack}.
\end{proof}

These corollaries lead to our algorithmic result under the cluster expansion criterion.
The following theorem subsumes \Theorem{cluster-no-slack} and adds a statement under the assumption of slack.

\begin{theorem}
\TheoremName{cluster-with-slack}
Let $E_1,\ldots,E_n$ be events and let $p_i = \Pr_\mu[E_i]$.
Suppose that the three subroutines described in \Section{algass} exist.
If $p$ satisfies \eqref{eq:CLL} then,
with probability at least $1-e^{-t}$, \MSR resamples no more than
$ 4 (\sum_{i=1}^{n} y_i) (\sum_{j=1}^{n} \ln (1+y_j) + 1 + t) $ events.

If \eqref{eq:CLL} is satisfied with a slack of $\epsilon \in (0,1)$,
i.e., $(1+\epsilon) p_i \leq y_i / Y_{\Gamma^+(i)}$,
then with probability at least $1-e^{-t}$, \MSR resamples no more than
$ \frac{2}{\epsilon} (\sum_{j=1}^{n} \ln (1 + y_j) + t) $ events.
\end{theorem}

\begin{proof}
The first statement follows directly from \Theorem{Shearer-no-slack},
since \Corollary{pInS} shows that $p \in \cS$
and \Corollary{qAndy} shows that $\frac{q_{\set{i}}}{q_\emptyset} \leq y_i$.
Next assume that \eqref{eq:CLL} is satisfied with $\epsilon$ slack.
We apply \Corollary{pInS} and \Corollary{qAndy} to the point $p' = (1+\epsilon) p$, 
obtaining that $p' \in \cS$ and $q'_{\set{j}} / q'_\emptyset \leq y_j$,
where $q'_S$ denotes $q_S(p')$.
The second statement then follows directly from \Corollary{Shearer-no-q0}. 
\end{proof}

\subsubsection{Proof of \Lemma{CLLimpliesShearer}}
\SectionName{CLLimpliesShearer}

\begin{claim}[The ``fundamental identity'' for $Y$]
\ClaimName{fundamentalY}
$Y_{A} = Y_{A - a} + y_a Y_{A \setminus \Gamma^+(a)}$
for all $a \in A$.
\end{claim}
\begin{proof}
Every summand $y^J$ on the left-hand side either appears in $Y_{A - a}$
if $a \not\in J$, or can be written as $y_a \cdot y^B$ where
$B = J \setminus \Gamma^+(a)$, in which case it appears
as a summand in $y_a Y_{A \setminus \Gamma^+(a)}$.
\end{proof}

\begin{claim}[Log-subadditivity of $Y$]
\ClaimName{submult}
$Y_{A \union B} \leq Y_A \cdot Y_B$ for any $A, B \subseteq [n]$.
\end{claim}
\begin{proof}
It suffices to consider the case that $A$ and $B$ are disjoint,
as replacing $B$ with $B \setminus A$ decreases the right-hand side
and leaves the left-hand side unchanged.
Every summand $y^J$ on the left-hand side can be written as
$y^{J'} \cdot y^{J''}$ with $J' = J \intersect A$ and $J'' = J \intersect B$.
The product $y^{J'} \cdot y^{J''}$ appears as a summand on the right-hand side,
and all other summands are non-negative.
\end{proof}

\begin{proofof}{\Lemma{CLLimpliesShearer}}
We proceed by induction on $|S|$. The base case is $S = \{a\}$. In that case we have
$ \frac{\qdown_{\{a\}}}{\qdown_\emptyset} = \qdown_{\{a\}} = 1 - p_a.$
On the other hand, by the two claims above and \eqref{eq:CLL}, we have 
$$ Y_{[n]} ~=~ Y_{[n]-a} + y_a Y_{[n] \setminus \Gamma^+(a)}
 ~\geq~ Y_{[n] - a} + p_a Y_{\Gamma^+(a)} Y_{[n] \setminus \Gamma^+(a)}
 ~\geq~ Y_{[n] - a} + p_a Y_{[n]}.$$
Therefore, $\frac{Y_{[n]-a}}{Y_{[n]}} \leq 1 - p_a$ which proves the base case.

We prove the inductive step by similar manipulations. By \Claim{breve1}, we have
$$ \frac{\qdown_S}{\qdown_{S-a}} ~=~ 
1 - p_a \frac{\qdown_{S \setminus \Gamma^+(a)}}{\qdown_{S-a}}.$$
The inductive hypothesis applied repeatedly to the elements of $S \cap \Gamma(a)$ yields
\begin{equation*}
1 - p_a \frac{\qdown_{S \setminus \Gamma^+(a)}}{\qdown_{S-a}}
 ~\geq~ 1 - p_a \frac{Y_{(S \setminus \Gamma^+(a))^c}}{Y_{(S-a)^c}}
 ~=~ 1 - p_a \frac{Y_{S^c \cup \Gamma^+(a)}}{Y_{S^c + a}}. 
\end{equation*}
By the two claims above and \eqref{eq:CLL}, we have
$$ Y_{S^c+a} ~=~ Y_{S^c} + y_a Y_{S^c \setminus \Gamma^+(a)}
 ~\geq~ Y_{S^c} + p_a Y_{\Gamma^+(a)} Y_{S^c \setminus \Gamma^+(a)}
 ~\geq~ Y_{S^c} + p_a Y_{S^c \cup \Gamma^+(a)}.$$
We conclude that
$$ \frac{\qdown_{S}}{\qdown_{S-a}}
 ~\geq~ 1 - p_a \frac{Y_{S^c \cup \Gamma^+(a)}}{Y_{S^c+a}}
 ~\geq~ 1 - \frac{Y_{S^c+a} - Y_{S^c}}{Y_{S^c+a}} = \frac{Y_{S^c}}{Y_{(S-a)^c}}.$$
\end{proofof}

\subsubsection{Relationship between cluster expansion and stable set sequences}
\SectionName{CLLSSS}

We remark that the following more general bound holds: For every $J \in \Ind$,
\begin{equation}
\EquationName{CLL-psum}
 \sum_{\cI \in \Prop(J)} p_\cI ~=~ \frac{q_J}{q_\emptyset} ~\leq~ y^J.
\end{equation}
The equality holds by \Lemma{Shearer-sum-equality} and the inequality
can be derived from \Lemma{CLLimpliesShearer} as follows:
$$ \frac{q_J}{q_\emptyset}
    ~=~ \frac{p^J \qdown_{(\Gamma^+(J))^c}}{\qdown_{(\emptyset)^c}} 
    ~\leq~ p^J \frac{Y_{\Gamma^+(J)}}{Y_\emptyset}
    ~=~ p^J Y_{\Gamma^+(J)} 
    ~\leq~ \prod_{j \in J} (p_j Y_{\Gamma^+(j)})
    ~\leq~ y^J $$
using \Claim{breve3} for the first equality,
and \Claim{submult} and \eqref{eq:CLL} in the last two inequalities.

A direct proof that $\sum_{\cI \in \Prop(J)} p_\cI \leq y^J$
can be obtained by an inductive argument similar to the 
proof of \eqref{eq:stab-induction} in \Section{LLLslack}.
An application of \Lemma{Shearer-sum-equality} then establishes \eqref{eq:CLL-psum}.
Earlier versions of this paper used this approach to relate the cluster expansion criterion
and Shearer's lemma.
Our new approach in \Corollary{pInS} has the advantage that it does not require
the limiting arguments used in \Lemma{Shearer-sum-equality}.

\section{Conclusions}

We have shown that the Lov\'asz Local Lemma can be made algorithmic in the abstract framework of
\ros. This framework captures the General LLL as well as Shearer's Lemma in the existential sense, and leads to efficient algorithms for the primary examples of probability spaces and events satisfying lopsidependency that have been considered in the literature (as surveyed in \cite{LuMohrSzekely}).
 
Our algorithmic form of the General LLL (\Theorem{LLL-tight-result}) uses 
$O\big(\sum_{i=1}^{n} \frac{x_i}{1-x_i} \sum_{j=1}^{n} \log \frac{1}{1-x_j} \big)$
resampling operations, which is roughly quadratically worse than the
$\sum_{i=1}^{n} \frac{x_i}{1-x_i}$ bound of Moser-Tardos \cite{MoserTardos}.
Similarly, our algorithmic result under Shearer's condition (\Theorem{Shearer-no-slack})
uses $O\big( \sum_{i=1}^{n} \frac{q_{\set{i}}}{q_\emptyset} \sum_{j=1}^{n} \ln (1 +
\frac{q_{\{j\}}}{q_\emptyset}) \big)$ resampling operations,
which is roughly quadratically worse 
than the $\sum_{i=1}^{n} \frac{q_{\set{i}}}{q_\emptyset}$ bound of Kolipaka-Szegedy
\cite{Kolipaka}.
Can this quadratic loss be eliminated?

One way to prove that result would be to prove an analog of the witness tree lemma,
which is a centerpiece of the Moser-Tardos analysis \cite{MoserTardos}.
The witness tree lemma has other advantages, for example in deriving
parallel and deterministic algorithms. Unfortunately, the witness tree lemma is not true in the
general setting of \ros (see \Appendix{witness-trees}).
It is, however, true in the variable model \cite{MoserTardos}
as well as in the setting of random permutations \cite{HarrisS14}. 
Is there a variant of our framework in which the witness tree lemma is true,
and which continues to capture the LLL in full generality?

\section*{Acknowledgements}

We thank Mohit Singh for discussions at the early stage of this work.
We thank David Harris for suggesting the results of \Section{product-resampling},
and for discussions relating to \Appendix{witness-trees}.

\appendix

\clearpage
\section{A counterexample to the witness tree lemma}
\AppendixName{witness-trees}

A cornerstone of the analysis of Moser and Tardos \cite{MoserTardos} is the {\em witness tree
lemma}. It states (roughly) that for any tree of events growing {\em backwards in time} from a
certain root event $E_i$, with the children of each node $E_{i'}$ being neighboring events resampled
before $E_{i'}$, the probability that this tree is consistent with the execution of the algorithm is
at most the product of the probabilities of all events in the tree. (We give a more precise
statement below.) Extensions of this lemma have been crucial in the work of Kolipaka-Szegedy on
algorithmic forms of Shearer's Lemma \cite{Kolipaka} and work of Harris-Srinivasan on the algorithmic local lemma for permutations \cite{HarrisS14}. The witness tree lemma leads to somewhat stronger quantitative bounds than the ones we obtain, and it has been also useful for other purposes: derandomization of LLL algorithms \cite{MoserTardos,CGH}, parallel algorithms \cite{MoserTardos,Pettie}, and handling exponentially many events \cite{HSS}. Therefore, it would be desirable to prove the witness tree lemma in our general framework of \ros.

Unfortunately, this turns out to be impossible. The main purpose of this section is to show that the
witness tree lemma is false in the framework of \ros in a strong sense. Whereas in typical scenarios
the Moser-Tardos algorithm only requires witness trees of depth $O(\log n)$ with high probability,
in the \ro framework the stable set sequences (and an analogous notion of witness trees)
can have nearly-linear length with constant probability.

Before we proceed, we define a few notions necessary for the formulation of the witness tree lemma. Our definitions here are natural extensions of the notions from \cite{MoserTardos} to the setting of \ros.

\begin{definition}
Given a \NAG $G$ on vertex set $[n]$,
a witness tree is a finite rooted tree $T$,
with each vertex $v$ in $T$ given a label $\cE_v \in [n]$,
such that the children of a vertex $v$ receive labels from $\Gamma^+(\cE_v)$.
\end{definition}

\begin{definition}
We say that a witness tree $T$ with root $r$ appears in the log of the algorithm, if event $\cE_r$
is resampled at some point and the tree is produced by the following procedure: process the
resampled events from that point backwards, and for each resampled event $j$ such that $j \in
\Gamma^+(\cE_v)$ for some $v$ in the tree, pick such a vertex $v$ of maximum depth in the tree and
create a new child $w$ of $v$ with label $\cE_w = j$.
\end{definition}

The witness tree lemma, in various incarnations, states that the probability of a witness tree $T$
appearing in the log of an LLL algorithm is at most $\prod_{v \in T} \Pr[E_{\cE_v}]$. We show here that this can be grossly violated in the setting of \ros. Our example actually uses the independent variable setting but resampling oracles different from the natural ones considered by Moser and Tardos.

\paragraph{Example.}  
Consider independent Bernoulli variables $X_i, Y_i^j, Z_i$ and $W$ where $1 \leq i \leq k$ and $1 \leq j \leq \ell$. The probability distribution $\mu$ is uniform on the product space of these random variables. Consider the following events:
\begin{compactitem}
\item $E_i = \set{ X_i = 0 }$
\item $E_i^j = \set{ Y_i^j = 0 }$
\item $E' = \set{ W = 1 }$
\end{compactitem}
These events are mutually independent. However, let us consider a dependency graph $G$ where $E_i
\sim E_i^j$ for each $1 \leq i \leq k, 1 \leq j \leq \ell$; this is a conservative choice but
nevertheless a valid one for our events. (One could also tweak the probability space slightly so that neighboring events are actually dependent.) In any case, $E'$ is an isolated vertex in the graph.

We define resampling oracles as follows. In the following, $Q$ describes a fresh new sample of a Bernoulli variable. Only the variables relevant to the respective oracle are listed as arguments.
\begin{compactitem}
\item $r_i(X_i) = Q$
\item $r_i^j(X_i, Y_i^j,Z_i) = (Z_i, Q, X_i)$
\item $r'(W,Z_1,\ldots,Z_k) = (Z_1,\ldots, Z_k,Q)$.
\end{compactitem}

\begin{claim}
$r_i, r_i^j,r'$ are valid resampling oracles for the events $E_i, E_i^j,E'$ and the dependency graph $G$.
\end{claim}

\begin{proof}
$r_i$ resamples only the variable $X_i$ relevant to event $E_i$ and hence cannot cause any other
event to occur. Conditioned on $E_i = \set{X_i = 0}$, it clearly produces the uniform distribution.

$r_i^j$ switches the variables $X_i$ and $Z_i$ and thus can cause $E_i$ to occur (which is
consistent with the dependency graph $G$). Conditioned on $E_i^j = \set{Y_i^j = 0}$, it makes $Y_i^j$ uniformly random and preserves a uniform distribution on $(X_i,Z_i)$.

$r'$ affects the values of $W,Z_1,\ldots,Z_k$ but no event depends on $Z_1,\ldots,Z_k$, so $r'$
cannot cause any event except $E'$ to occur. Conditioned on $E' = \set{W=1}$, since $(Z_1,\ldots,Z_k)$ are distributed uniformly, it produces again the uniform distribution. 
\end{proof}

\paragraph{The Moser-Tardos algorithm.}
First, let us consider the Moser-Tardos algorithm: In the most general form, it resamples in each step an arbitrary occurring event. For concreteness, let's say that the algorithm always resamples the occurring event of minimum index (in some fixed ordering).

\begin{claim}
If the Moser-Tardos algorithm considers events in the order $(E_i, E_i^j, E')$, then at the time it gets to resample $E'$, the variables $Z_1,\ldots, Z_k$ are independent are equal to $1$ with probability $1 - 1/2^{\ell+1}$ each.
\end{claim}

\begin{proof}
Let us fix $i$. Whenever some variable $Y_i^j$ is initially equal to $0$, we have to resample $E_i^j$ at some point. However, we only resample $E_i^j$ if $E_i$ does not occur, which means that $X_i$ must be $1$ at that time. So the resampling oracle $E_i^j$ forces $Z_i$ to be equal to $1$. The only way $Z_i$ could remain equal to $0$ is that it is initially equal to $0$ and none of the events $E_i^j$ need to be resampled, which happens with probability $1/2^\ell$. Therefore, when we're done with $E_i$ and $E_i^j$ for $1 \leq j \leq \ell$, $Z_i$ is equal to $0$ with probability $1 / 2^{\ell+1}$. This happens independently for each $i$.
\end{proof}

\begin{lemma}
The probability that the Moser-Tardos algorithm resamples $E'$ at least $k$ times in a row is at least $\frac12 (1 - \frac{1}{2^{\ell+1}})^{k-1}$.
\end{lemma}

\begin{proof}
By the ordering of events, $E'$ is resampled only when all other events have been fixed. Also, resampling $E'$ cannot cause any other event, so the algorithm will terminate afterwards. However, as we argued above, when we get to resampling $E'$, each variable $Z_i$ is equal to $1$ independently with probability $1 - 1/2^{\ell+1}$. Considering the resampling oracle $r'(W,Z_1,\ldots,Z_k) = (Z_1,\ldots,Z_k,Q)$, if $W$ as well as all the variables $Z_i$ are equal to $1$, it will take at least $k$ resamplings to clear the queue and get a chance to avoid event $E'$. This happens with probability $\frac12 (1 - \frac{1}{2^{\ell+1}})^{k-1}$.
\end{proof}

Let $T$ consist of  a path of $k$ vertices labeled $E'$.
For $k = 2^\ell$, we conclude that the witness tree $T$ appears with constant probability in the log of the Moser-Tardos algorithm, as opposed to $1/2^{k}$ which would follow from the witness tree lemma.

\paragraph{The \MSR algorithm.}
A slightly more involved analysis is necessary in the case of \MSR. By nature of this algorithm, we would resample $E'$ ``in parallel" with the other events and so the variables evolve somewhat differently. 

\begin{claim}
For each $i$ independently, after 2 iterations of the \MSR algorithm, $Z_i = 1$ with probability
$1-1/2^{\ell+1}$. Any further updates of $Z_i$ other than those caused by resampling $E'$ can only change the variable from $0$ to $1$.
\end{claim}

\begin{proof}
The claim is that unless $Z_i = 0$ and $Y_i^1 = \ldots, = Y_i^\ell = 1$ initially, in the first two
iterations we will possibly resample $E_i$ and then one of the events $E_i^j$, which makes $Z_i$
equal to $1$. Any further update to $Z_i$ occurs only when $E'$ is resampled (which shifts the sequence $(Z_1,\ldots,Z_k)$) or when $E_i^j$ is resampled, which makes $Z_i$ equal to $1$.
\end{proof}

\begin{lemma}
The probability that \MSR resamples $E'$ at least $k$ times in a row is at least $\frac14 (1 - \frac{1}{2^{\ell+1}})^{k-2}$.
\end{lemma}

\begin{proof}
In the first two iterations, the probability that $E'$ is resampled twice is at least $1/4$ (the values of $W$ and $Z_1$ are initially uniform, and if $Z_1$ is updated, it can only increase the probability that we resample $E'$). Independently, the probability that $Z_2 = \ldots = Z_{k-1} = 1$ after the first two iterations is $(1 - 1/2^{\ell+1})^{k-2}$, by the preceding claim. (We are not using $Z_1$ which is possibly correlated with the probability of resampling $E'$ in the second iteration, and $Z_k$ which would be refreshed by this resampling in the second iteration.) If this happens, we will continue to resample $E'$ at least $k-2$ additional times, because it will take $k-2$ executions of $r'$ before a zero can reach the variable $W$.
\end{proof}

Again, consider setting $k=2^\ell$.
The total number of events is $n = O(k \ell)$,
so $\ell = \Theta(\log n)$ and $k = \Theta(n / \log n)$.
With constant probability, the witness tree $T$ consisting of a path of $k$ vertices labeled $E'$
will appear in the log of \MSR algorithm.
Thus, with constant probability, the algorithm will require a stable set sequence of length at least
$k$.

\end{document}